\documentclass[11pt]{article}
\usepackage[letterpaper, margin=1.2in]{geometry}

\usepackage{macros}
\usepackage{stackengine}
\usepackage{centernot}

\newcommand{\domain}{\mathbb D}

\newcommand{\lincomb}{\operatorname{Lin}}

\newcommand{\fsfun}[2]{ #1 \underset{\text{fs}}{\longrightarrow} #2}

\newcommand{\leftderivative}[2]{#1(#2\underline{\hspace{2mm}})}

\newcommand{\termop}{ \underset{\text{term}}{\longrightarrow}}

\input{knowledges.kl}

\begin{document}

\title{Regularity as seen by Alice and Bob}
\author{Miko{\l}aj Boja\'nczyk, Aliaume Lopez, Rafa{\l} Stefa\'nski, Omid Yaghoubi }

\maketitle 
\begin{abstract}
The goal of this paper is to propose a unifying model for Nerode-style
characterizations of regularity across functions with different output
domains. Building on Hauser's work in communication complexity, we
generalize the setting by relaxing the computability assumptions and
allowing non-Boolean output domains. We consider functions of type
$\Sigma^* \to \domain$, where $\Sigma$ is a finite alphabet and
$\domain$ is an arbitrary domain. For several domains, we show that
the model coincides with known models of computation. We further
conjecture that an analogous correspondence holds for other domains
that currently lack a Nerode-style characterization of regularity,
and we provide ample supporting evidence. In the model, an input
string $w$ is split as $w = w_1 w_2$ and distributed between two
cooperating parties, Alice and Bob, who exchange a constant number
of messages to compute the value of the function. Each message is
either an element of the output domain or a signal drawn from a
finite set of signals, and the parties must produce the correct
output for every admissible split $w = w_1 w_2$. We further extend
the framework to infinite alphabets in the setting of nominal sets,
and investigate its expressiveness on languages of words with atoms.
\end{abstract}

\newpage

\section{Introduction}
\label{sec:introduction}

\AP
This paper is motivated by a desire to understand the notion of  regularity in formal language theory. We take the functional perspective, in which we consider functions 
\begin{align*}
f : \Sigma^* \to \domain
\end{align*}
that input strings, and output values from some domain $\domain$. If the output
domain is the Booleans, then such functions are languages, and there is no
question about which languages should be considered \intro[regular language]{regular}. There are tens --
if not hundreds -- of equivalent definitions, including regular expressions,
finite automata in numerous forms, monoids, monadic second-order logic, and
variants of $\lambda$-calculus. But what about other outputs? Let us review
three examples where the nature of regularity is a topic of genuine
discussion.

\begin{enumerate}
    \item \textbf{String outputs.}
Consider string-to-string functions 
\begin{align*}
f : \Sigma^* \to \Gamma^*.
\end{align*}
Similarly to languages, the literature on automata theory offers countless  models. This time, however, not all of them are equivalent, but there is at least some semblance of order. Let us mention three classes of functions  of particular interest:  the \emph{rational}, \emph{regular}, and \emph{polyregular} functions. These classes are described in~\cref{fig:transducer-classes} in Section~\ref{sec:string-outputs} together with the appropriate references, with each one having at least five different characterisations, using models of varied origins, including logic, algebra and programming language theory. Which of these classes, if any, should be considered ``the'' regular string-to-string functions? We could simply go with the middle one, because  the word ``regular'' is traditionally used for it, but a more principled approach would be preferable.

\item \textbf{Number outputs.}
Consider string-to-number functions, say functions 
\begin{align*}
f : \Sigma^* \to \Rat
\end{align*}
that output rational numbers (more generally, the outputs could be in some field). Here, the literature offers two natural candidates, namely \emph{weighted automata}~\cite{schutzenberger1961definition}, and \emph{polynomial automata}~\cite[Section IV]{DBLP:conf/lics/BenediktDSW17}. In both cases, there is an automaton that reads the input string in one pass, and stores in its state a vector $\Rat^d$ of some fixed dimension. In weighted automata, this vector is updated using linear maps, while polynomial automata can use polynomial maps. These models are not equivalent -- polynomial automata are strictly more powerful -- but both have a good mathematical theory, one based on linear algebra, and the other based on algebraic geometry. Again, we might be tempted to ask: which  is the right one?

\item \textbf{Infinite input alphabets.}
Our final example is of a different kind than the previous ones. We return to functions with Boolean outputs
\begin{align*}
f : \Sigma^* \to \set{\text{true, false}},
\end{align*}
i.e.~languages, but this time the input alphabet is no longer required to be finite. The infinity needs to be somehow tamed, and the standard approach to do this is to use an infinite alphabet where letters can only be compared for equality~\cite{kaminskiFiniteMemoryAutomata1994}. This allows for languages such as ``all letters are different'' or ``the first letter is equal to the last one'', but not for languages such as ``the letters are in increasing order'', since the letters are not equipped with a linear order, or any other kind of structure.
The literature is rife with automata models for such languages, with seventeen examples listed in \cref{fig:automata-infinite-alphabets}, all describing pairwise non-equivalent models. 
 Again, we might be tempted to ask: which  is  the right one?
\end{enumerate}

This type of question can be asked in other settings, with   outputs such as  trees,  graphs or elements of some abstract semiring. One could also vary the inputs, and consider regularity for, say, graph-to-graph functions, but we refrain from considering general outputs in this paper, and we stay with string inputs. This paper attempts to provide a unified theory of regularity for such functions.   We are guided by  the following principle, which we believe to be essential for regularity:
\begin{description}
    \item[Constant information flow.]   If the input is split into parts, then only a constant amount of information  flows between them, as far as the output of the function is concerned.
\end{description}
This principle is only a vague guideline, since it does not identify what
``information flow'' means, or how to quantify its amount. For Boolean outputs,
i.e.~languages, constant information flow has a standard interpretation, which is the
Myhill-Nerode Theorem, and it is known to
correspond exactly to the regular languages. However, in the case of more
complicated outputs, things are less clear.  For example, if the outputs of the
function are rational numbers, then it should be legitimate for the information
flow to contain some rational numbers. How should this be formalised?

\AP
We choose to use a formalisation that is based on communication complexity~\cite{YAO79,KUSH97}. In our model, the function is computed by two cooperating parties, called Alice and Bob. 
There are \textbf{no uniformity assumptions} and the two parties
have unrestricted computational power; the goal is to measure information 
and
not computation.  The input string is split into two parts $w = w_1 w_2$, and
Alice has access to $w_1$ while Bob has access to $w_2$. The two parties  exchange  a
\textbf{constant number of messages} in order to compute the function, where the constant  depends only on the
protocol and not on the input string. The  output of the computation must
be \textbf{split invariant}, which means that the output depends only on the
input string $w$, and not on the split $w = w_1 w_2$. In the communication,
there are two kinds of messages:  either bits in $\set{\text{true, false}}$  or
elements of the output domain (as far as we can tell, the messages from the domain are the novelty of our approach).
For example, if the outputs are rational
numbers, then the messages can contain rational numbers. However, there is a
restriction on the access to messages from the output domain, which is called
the \emph{\textbf{\intro{black box restriction}}}. This restriction (which will be formally
defined later in the paper) is intended to prevent tricks such as Alice sending
her input string to Bob encoded as a rational number. Intuitively speaking,
the  black box restriction says that the parties cannot read the messages from
the output domain, and instead they can only operate on them using predefined
operations. For example, in the case of numbers, the operations are addition
and multiplication.

The model, which we call \emph{protocols}, can be applied to any output domain,
and we study several examples in this paper. In the case of Boolean outputs, the model and its connection to finite automata have already been studied before~\cite[Theorem 5]{hauser1989}, however the results on other output domains are new up to our best knowledge.   As we discover, despite its
non-uniformity, the model  can only define well-behaved functions. In
particular, in all cases that we have studied  these functions are: (1)  always
computable, even in linear time; and (2) the output size is always at most
linear (with the size of a rational number measured by the number of bits
needed to represent it).  This seems to indicate that the split invariance,
together with a  constant number of messages under the black box restriction,
has unexpected computational consequences. In particular, in all cases that we
have studied, protocols coincide -- or are conjectured by us to coincide --
with existing automata models. Since automata are conceptually very different 
than protocols, we believe that those equivalences, summarized in the table below,
justify the protocol-based approach to regularity.

\begin{center}
    \begin{tabular}{ll}
    \textbf{Output} & \textbf{Automata Model} \\
    \hline
    \kl[Boolean domain]{Booleans} & Finite automata (\cref{thm:boolean-domain}) \\
    \kl[String domain]{Field}   & Weighted automata (\cref{thm:field-domain}) \\
    \kl[String domain]{Strings} & Two-way automata with output (\cref{conj:protocol-regular-string-to-string}) \\
    Boolean, but infinite input alphabet & Unambiguous automata (\cref{conj:protocols-unambiguous}) 
\end{tabular}
\end{center}

In the remainder of this introduction, we give a more detailed review of the four rows
in the table, including substantial evidence for the conjectured equivalences.

\subsection{Protocols for Boolean outputs}
\label{sec:intro-boolean}

We begin by studying the protocol model for  functions
\begin{align*}
f : \Sigma^* \to \set{\text{true, false}},
\end{align*}
i.e.~languages. This case is not new, and it has already been studied in~\cite{hauser1989}. For Boolean outputs, the two parties only exchange bits, 
and they need to decide if the input string is in the language or not.
Here is an example.

\begin{myexample}[Parity]
    \label{ex:three-letters}
Suppose that the language is ``the string has even length''. In a protocol for this language, Alice sends the parity for her part of the input (one bit), and Bob uses this bit to return the final answer.  
\end{myexample}

As mentioned before, the parties have unbounded computational power, which means that  their messages
can contain answers to potentially undecidable questions about their parts of the input.
Despite this, the protocols compute exactly the regular languages.
This was first shown in~\cite[Theorem 5]{hauser1989}, but we restate it here for completeness, and provide a self-contained proof later in the paper.

\begin{restatable}{theorem}{booleandomain}
 \label{thm:boolean-domain} A language $L \subseteq \Sigma^*$ is computed by a \kl[boolean protocol]{protocol} if and only if it is 
 \kl[regular language]{regular}.
\end{restatable}
One implication is immediate: every regular language can be computed by a protocol.  Alice can send to Bob the state of a finite automaton that recognises the language. The other implication is proved in two steps, see Section~\ref{sec:boolean-domain} for more details. In the first step, the protocol is reduced to a one-round non-interactive version, where each party independently sends a message with a constant number of bits, and the decision is then made based on these two messages. In the second step,  the Myhill-Nerode Theorem is used to show that the language must necessarily be regular. 

Theorem~\ref{thm:boolean-domain} is  simple  technically, and its main value for us lies in its role as an inspiration for other results, which use  other output domains  such as  fields or strings.


Before moving to the other output domains, let us comment on other related work that connects communication complexity with automata in the case of Boolean outputs.  Much of this work is related to \emph{state complexity}, where one studies  the number of states needed for a given language in a given automaton model, and how this number is affected by operations on languages or changes in the model. See  Wikipedia~\cite{stateComplexityWiki} for a comprehensive summary with numerous references, or the recent paper~\cite{goosKiefer2022} which shows how to transfer lower bounds from communication complexity to state complexity of unambiguous automata. 
 The work on state complexity is mainly about the exact number of states, which is of secondary concern to us. For our purposes, a protocol that exchanges $k$ bits is no different from a protocol that exchanges $2^k$ bits. We only care that this bound should be finite and independent of the input.

\subsection{Field outputs}
\label{sec:intro-field}

After the Booleans, we turn to functions with outputs in a field. This is the first original contribution of this paper. For the sake of concreteness, let us  consider functions with outputs in the  field of rationals
\begin{align*}
f : \Sigma^* \to (\mathbb Q, +, \times).
\end{align*}
We adapt the  protocol model to  compute such functions. Similarly to the Boolean case, the parties exchange messages. However, this time there are two kinds of messages: bits (as in the case of  Boolean outputs), and  elements of the field.  The elements of the field can be added and multiplied.  Division is not allowed, and its role is discussed in Example~\ref{ex:division}. Using bit messages, we can still recognise all regular languages (more formally their characteristic functions).
However, messages from the field allow computing new functions.

\begin{myexample}[Length and exponential length]\label{ex:length}
    Consider the following two functions
\begin{align*}
\myunderbrace{w \mapsto |w|}{length} \qquad \text{and} \qquad \myunderbrace{w \mapsto 2^{|w|}}{exponential length}.
\end{align*}
To compute the length of the input string, Alice  sends the length of her part, and Bob adds this to his length, thus yielding the desired output. For the exponential length, we use a similar protocol, except that multiplication is used instead of addition. 
\end{myexample}

In the presence of an infinite message space, one needs to be careful about the design of the protocol. For example, Alice could try to send her part of the input string in a single message by encoding it as a rational. This would trivialise the model, enabling every function to be computed. To prevent such tricks, we use the \emph{black box restriction} which was discussed before\footnote{The black box postulate is related to  polymorphic parametricity from the theory of programming languages~\cite[Section 7]{reynolds1983types}, or to the recent algebraic group model in cryptography~\cite[Section 1.2]{fuchsbauer2018algebraic}. An important difference with the algebraic group model is that our model does not allow for equality tests, see Example~\ref{ex:equality-tests} for a discussion.}: the messages which use the output domain (in this case, rational numbers) cannot be read directly, but can only be acted on by the operations in the output domain (in this case addition and multiplication).
 If the output domain is finite, e.g.~it is a finite field, then the black box restriction is irrelevant (which is why it was not mentioned when talking about Boolean outputs). This is because one can use bits to sent elements of a finite output domain, and the bits are a preferrable communication channel, since they can be read directly and are not subject to the black box restriction. 

\paragraph*{Definition of the protocol model.} Since there might be some ambiguity as to what exactly is allowed in the protocol, we give a more formal  definition. There is a finite input alphabet $\Sigma$, and a constant number of rounds $k \in \set{1,2,\ldots}$.   Alice and Bob send messages in alternation, with Alice sending the first message\footnote{One could consider other patterns of communication. In fact, in Section~\ref{sec:beyond-boolean-outputs} we use a more symmetric variant where both parties move in parallel in each round. These variants  do not change the expressive power of the model, only the  number of rounds.}. The messages are from 
\begin{align*}
\myunderbrace{\set{\text{true, false}} + \Rat}{disjoint union of bits and rational numbers}.
\end{align*}
When choosing their message in the $i$-th round, the corresponding  party (Alice in odd rounds, Bob in even rounds) has access to their part of the input string, and the  bits from previous messages. The numbers  are rationals, and cannot influence the decision, as per the black box restriction.  The information available  in the $i$-th round is given by  the set 
\begin{align}\label{eq:strategy-input}
\myunderbrace{\Sigma^*}{part of the  input \\ string that is \\ known to the party}
\qquad \times \qquad 
\myunderbrace{\set{\text{true, false, unknown}}^{i-1}}{messages received in  previous rounds, \\ with  numbers from $\Rat$ replaced by ``unknown''}.
\end{align}
Based on this information, the corresponding party chooses a new message to  send, which is either a bit, or a number. The number can be produced in two ways: either  a  fresh number  is produced based on the available information, or otherwise two previously received numbers are combined using addition or multiplication.  Therefore, the possibilities for the message sent in the $i$-th round are described by the set 
\begin{align}
    \label{eq:strategy-output}
\myunderbrace{\set{\text{true, false}}}{bits}
\ + \ 
\myunderbrace{\Rat}{fresh \\ number}
\  + \ 
\myunderbrace{\setbuild{(op,x,y)}{$op \in \set{+,\times}$ and $x,y \in \set{1,\ldots,i-1}$}}{addition or multiplication of previously received numbers}.
\end{align}
If addition or multiplication is used, then the party sending the message is responsible for the operation to be well-defined, i.e.~the messages sent in rounds $x$ and $y$ must have been numbers.
Summing up, the strategy in round $i$ is a function which inputs an element of the set from \eqref{eq:strategy-input}, and outputs an element of the set from \eqref{eq:strategy-output}. This function need not be computable. The protocol is then described by $k$ such strategies, one for each round $i \in \set{1,\ldots,k}$. We assume that the last message, sent in the $k$-th round, is a number and not a bit -- this number is defined to be the output of the protocol. 
Finally, the protocol must be \kl{split invariant}, i.e.~for every input string $w$, the same output must be produced regardless of the factorization $w = w_1 w_2$ into strings owned by Bob and Alice. This completes the definition of the protocol model, in the case of field outputs.

\paragraph*{Equivalence with weighted automata.} Our main result for field outputs is~\cref{thm:field-domain} below, which says that protocols are equivalent to weighted automata. The precise definition of weighted automata will be given later in Section~\ref{sec:field-domain}, but the rough idea is that a weighted automaton maintains a vector  $\domain^d$ of field elements, with each input letter updating the vector via some fixed linear map.

\begin{restatable}{theorem}{fielddomain}
    \label{thm:field-domain}
     Assume that the domain $\domain$ is a field. Then a function 
    \begin{align*}
    f : \Sigma^* \to \domain
    \end{align*}  is computed by a \kl{protocol} if and only if it is  computed by a \kl{weighted automaton} over the same field.
\end{restatable}

  This result might even seem surprising. Our protocol is designed to use polynomial operations on the output domain, and therefore one could expect the relevant automaton model to be also based on polynomials, such as the  polynomial automata of~\cite{DBLP:conf/lics/BenediktDSW17}, which are the extension of weighted automata that allow polynomial maps instead of linear ones. As it turns out, the split invariance in the protocols  enforces linearity, and thus it excludes the general polynomials operations that are used in polynomial automata. The linearity phenomenon is true for outputs in a field -- because weighted automata are based on linear maps -- and it will also be true for other output domains, such as strings, see \cref{thm:evidence-for-the-conjecture}. We do not have a fully general understanding of this phenomenon.

  The proof of \cref{thm:field-domain} is given in Section~\ref{sec:proof-of-thm-field-domain}, but here we present a rough outline.  The proof is similar to the one used for the case of Boolean outputs, and has  two steps:

\begin{enumerate}
  \item We first show in Section~\ref{sec:reduction-to-scalar-product-protocols} that any protocol with outputs in a field can be reduced to a special  form, which we call a \emph{\kl{scalar product protocol}}. In this protocol, Alice and Bob apply in parallel two functions 
\begin{align*}
\sigma_A, \sigma_B : \Sigma^* \to \Rat^d
\end{align*}
to their parts of the input string, where $d$ is some fixed finite dimension. Then, the output is obtained by combining these two  $d$-dimensional vectors using scalar product. A scalar product protocol can be simulated by the general version of the protocol, but it is subject to certain restrictions: (a) there is no interaction; and (b) bits are not used, only  field elements. Since, as we prove, every product can be reduced to this scalar form, it follows that interaction and bit messages are not needed in the protocol. In fact, the interaction can be removed for all output domains, but the removal of bits is specific to fields. 
\item After reducing to the scalar product protocols, the next step (see Section~\ref{sec:from-scalar-product-protocol-to-weighted-automaton}) is to apply a version of the Myhill-Nerode Theorem for weighted automata. This is called the  Fliess Theorem~\cite{fliess1974}, and it says that recognisability by a weighted automaton is equivalent to having finite rank for a certain matrix, which is called the Hankel matrix. Roughly speaking, the rows in the Hankel matrix, in the context of our protocols, describe strategies of Alice, and the columns describe strategies of Bob. Therefore, it is not hard to show that in a scalar product protocol that uses vectors of dimension $d$,  the Hankel matrix has rank at most $d$. This, together with the Fliess Theorem, shows that protocols are equivalent to weighted automata, completing the proof of \cref{thm:field-domain}.
\end{enumerate}

\begin{myexample}[Division]\label{ex:division} What if we added division to the operations?  Consider the function  $w \mapsto 1/(|w|+1)$. This function can easily be computed using a protocol with division. We now argue that it cannot be computed using addition and multiplication only, thus proving that division gives extra power. Since the function depends only on the length of the input, it can be seen as having type $\Nat \to \Rat$. In such a type, weighted automata are the same as linear recurrence sequences. The inverse function $1/(n+1)$ is not a  linear recurrence sequence, which can be shown using the exponential polynomial form~\cite[Theorem 2.1]{BerstelReutenauer08}. Summing up, the choice of operations is important; we use a field, but we only allow the ring operations. We do not know what happens if division is allowed.
\end{myexample}

\paragraph*{Related work.} \cref{thm:field-domain} can be seen as a machine independent characterisation of  weighted automata. This would not be the first such characterisation, e.g.~the Fliess Theorem itself can be seen as a machine independent characterisation. Other research related to the Fliess Theorem is the categorical approach to  minimisation of weighted automata from~\cite{colcombetPetrisan2017}. We think that the value of our approach is that it places weighted automata in a broader context, which is defined purely in terms of communication, and in a way that is applicable to  other output domains, such as strings that will be considered next. As far as we know, the only  work  which takes such a broad view is the cost register automata of~\cite[Section C]{alurDantoniDeshmukhYuan2013}, which are an automaton model that describes functions with outputs in an arbitrary output domain, similarly to our setting. However, unlike our model, cost register automata are defined in terms of a  finite state machine model, and as such they lack the abstract machine independent flavour of our approach.

\subsection{String outputs}
\label{sec:intro-strings}

Our third group of results concerns string-to-string functions
\begin{align*}
f : \Sigma^* \to \Gamma^*.
\end{align*}
We  use the same kind of protocol as in the previous section, except that instead of numbers, the black box messages  contain strings from $\Gamma^*$, and instead of addition and multiplication, we have string concatenation.  Here are some examples.

\begin{myexample}[Reversal and duplication]\label{ex:reverse-duplicate}
    Using the protocol, we can compute string reversal: Alice sends to Bob the reverse of her part of the input, and  Bob concatenates this with  his part of the reverse. Another string-to-string function that can be easily computed by a protocol is string duplication $w \mapsto ww$. 
\end{myexample}

The string-to-string case is of particular interest because, as we have mentioned earlier in the introduction, there is no consensus as to  which string-to-string functions should be considered ``regular''. There are numerous automata models to choose from, some of which are summarized in \cref{fig:transducer-classes}, which contains twenty models, grouped by equivalence into three classes. We can exclude the weakest class (the rational functions), since it is too weak: it  cannot compute the reverse or duplicate functions from Example~\ref{ex:reverse-duplicate}. We can also exclude the strongest class (the polyregular functions), since it is too strong: polyregular functions can have superlinear output size, and as we will see in a moment, protocols can only have linear output size. By a process of elimination, we are left with the final  class from \cref{fig:transducer-classes}, which is traditionally called the ``regular functions''. One of the definitions of the class is in terms of   deterministic two-way automata with output~\cite{shepherdson1959reduction}; this model  is formally defined in Section~\ref{sec:string-outputs}. We conjecture  that this class     is the correct answer (thus validating the traditional name): 

\begin{conjecture}\label{conj:protocol-regular-string-to-string}
    A string-to-string function is computed by a protocol iff it is computed by a deterministic two-way automaton with output.
\end{conjecture}

In Section~\ref{sec:string-outputs} we discuss this conjecture in detail, and provide evidence in its favour, including a proof of the  implication 
\begin{align*}
\text{protocol} \quad \impliedby \quad \text{two-way automaton},
\end{align*} 
This implication is rather easy, since a two-way automaton can be neatly simulated by the repeated interactions of the protocol. The content of the conjecture, and the subject of the more technical results, is the $\implies$ implication. Most of Section~\ref{sec:string-outputs} is devoted to evidence for this implication. 
Our first argument  is the following result, which shows that functions computed  by protocols have many properties which are known to hold for deterministic two-way automata with output.

\begin{restatable}{theorem}{evidencefortheconjecture}
    \label{thm:evidence-for-the-conjecture}
    If a string-to-string function  is  computed by a \kl{protocol}, then:
    \begin{enumerate}
        \item \label{it:linear-size-outputs} outputs have at most linear size;
        \item \label{it:linear-time-computable} outputs can be   computed in linear time (ignoring logarithmic factors);
        \item \label{it:regular-preimages} preimages of regular languages are regular.
    \end{enumerate}
\end{restatable}

One can invent functions which satisfy the three conditions in the above
theorem, but which are not computed by deterministic two-way automata with output, see
Example~\ref{ex:not-regular-but-continuous-over-finite-fields}. However, all
known examples of such functions  are artificial, and none can be computed by
protocols (or any known transducer models).  The proof of
\cref{thm:evidence-for-the-conjecture}, which is given in
Section~\ref{sec:string-outputs}, uses a linear representation of strings as
matrices, and then applies \cref{thm:field-domain} about protocols with field
outputs. In fact, this technique suggests an alternative approach to regularity,
which connects string-to-string functions with the better understood case of
string-to-field functions. This approach is discussed in
Section~\ref{sec:continuity}.

Our second argument in favour of the \cref{conj:regular-continuous} is that we
can prove it in the special case when the output alphabet is unary, i.e.~it has
only one letter. This is the content of
Section~\ref{sec:unary-output-alphabet}. When the output alphabet is unary, the
output strings are commutative, i.e.~the order of letters is irrelevant; this commutativity is essential to our proof of the conjecture in the unary output case.  The proof is based on well-quasi orders, plus some results on weighted automata.

\paragraph*{Related work.} One of the consequences of the Myhill-Nerode theorem
for string-to-Boolean functions, or the Fliess Theorem for string-to-number
functions, is the existence of  canonical devices. There have been several
attempts to generalise this to string-to-string functions, with a special
emphasis the canonical devices. Before recalling this work, we observe that our
approach seems to go in a  different direction. Although we think of
Conjecture~\ref{conj:protocol-regular-string-to-string} as being a machine
independent characterisation, it does not necessarily  yield canonical devices.
In particular, our proof of the conjecture for unary alphabets does not yield a
canonical device.

Here is a summary of results on canonical devices for string-to-string
transducers: they have been proposed for subsequential functions~\cite[Théorème
1.1]{choffrut1977}, rational functions in~\cite[Theorem
1]{reutenauerSchutzenberger1991},  and for rational functions on infinite
words~\cite[Section 4]{canonicalRational2018}. A certain drawback of this line
of work is that: (a) the canonical devices are relative to a given automaton
model, which does not help in choosing one model over another; and (b) the
``canonical'' devices are not truly unique, since they depend on extra
parameters, such as the output delay for subsequential functions or the
lookahead for rational ones. Let us now move to the larger class of regular
functions, which is the subject of our conjecture. Here,  canonical model are
unknown, and the only known way to recover them uses non-standard semantics,
called origin semantics~\cite[Theorem
1]{bojanczykTransducersOriginInformation2014}. Another result of this kind,
which is a machine independent characterisation of the regular functions that
does not yield canonical devices, see \cite[Theorem
3.2]{bojanczykTitoRegular23}, is also implicitly based on origin semantics.
Finally, for the polyregular functions, the situation is of course even harder,
and the only known results concern a unary output alphabet~\cite[Section
IV]{Zpolyreg23}.

\subsection{Infinite input alphabets}
\label{sec:intro-infinite}
Our final group of results is about languages over infinite input alphabets. This is a departure from the previous setting, where  the input alphabet was always finite. Following the standard approach in automata theory, we assume that letters can only be compared for equality. Formally, we only want to consider languages that are invariant under permutations of the
alphabet, i.e.
\begin{align*}
w \in L 
\quad \iff \quad
\pi(w) \in L \qquad \text{for every permutation $\pi$ of the alphabet}.
\end{align*}
An example of such a language is ``all letters are different'', but not ``the letters are in the increasing order''.
 As mentioned earlier in the introduction, there is a rich literature on automata for such languages, see the surveys~\cite{neven2003power,segoufin2006automata,bojanczykOrbitFiniteSetsTheir2017} or the lecture notes~\cite{bojanczyk_slightly}.  The  relevant automata models typically use registers to store some letters from the input, so that they can be compared to later letters. Essentially any automaton model for finite alphabets can be lifted this way to infinite alphabets~\cite[Figure 1]{neven2003power}, and there is even a systematic way to do this, which is based on the theory of orbit-finite sets~\cite[Chapter 2]{bojanczyk_slightly}. Unfortunately, after this lifting, previously equivalent models become non-equivalent.  This sad situation is illustrated in~\cref{fig:automata-infinite-alphabets}, which describes seventeen non-equivalent automata models for infinite alphabets; all of these models collapse to the regular languages when considered for finite alphabets. Equally sadly, there are almost no results in the literature on infinite alphabets that prove non-trivial equivalences of models. The only known cases of this kind are about the  weakest of the available models, namely orbit-finite monoids~\cite{bojanczykNominalMonoids2013}, which are known to be equivalent in expressive power to a certain variant of \mso~\cite[Theorems 4.2 and 5.1]{DBLP:journals/corr/ColcombetLP14}, and also to single-use register automata~\cite[Theorem 6]{bojanczykstefanski2020}.
This situation desperately calls for some unifying principles. 

Since protocols have successfully identified important models in the previous
cases, we try to see  what happens in the case of an infinite input alphabet.
When extending protocols to infinite input alphabets, we adapt them as follows:
(1) messages can contain input letters; (2)  input letters can only be compared for
equality. Condition (2) is formalised by saying that the execution of the
protocol is invariant under permutations of the input alphabet, similarly to
the languages that we consider. The protocols  work their magic once again, and
they point to (as we conjecture)  one of the models in the literature. Before
revealing this model, let us briefly compare protocols to the two most prominent models 
for infinite alphabets.

\begin{myexample}[Deterministic too weak]\label{ex:reg-det-too-weak} We begin by considering a popular deterministic automaton model for infinite alphabets, namely i.e. \emph{deterministic register automata} \cite[Definition~3]{kaminskiFiniteMemoryAutomata1994}. This model is defined formally in Section~\ref{sec:infinite-alphabets}, but roughly speaking it is nondeterministic a finite automaton that is additionally equipped  with a fixed number of registers, which can be used to store input letters and compare them for equality with later input letters.  For example, the automaton can store the first input letter in a register, and then compare it with all later letters, thus recognising the language ``the first letter appears elsewhere in the word''. 

    A protocol can simulate any deterministic register automaton, using the same idea as for finite alphabets,
    i.e. Alice sends the intermediate state, including the register values, to Bob, who 
    continues the run on his side. To see that the inclusion is strict, consider the language
    ``the last letter appears elsewhere in the word'', which is the reverse of the language described in the previous paragraph. It is computed by a protocol, 
    in which Bob checks the condition on his side, and sends the last letter to Alice, so that she can verify
    that it does not appear on her side. On the other hand, it is a well-know fact that this language
    is not recognised by a deterministic register automata, as it proves that deterministic register automata 
    are not closed under reverse \cite[Examples~4~and~8]{kaminskiFiniteMemoryAutomata1994}. 
\end{myexample}

\begin{myexample}[Nondeterministic too strong]
    \label{ex:reg-ndet-too-strong} Consider the nondeterministic variant of the automata from the previous paragraph~\cite[Definition~1]{kaminskiFiniteMemoryAutomata1994}. This model is already too strong for the protocols, as attested by the language
    ``some letter appears twice'', which can be computed by an automaton, see~\cite[Example~1]{kaminskiFiniteMemoryAutomata1994}, but not by  a protocol. The intuitive reason for the negative result is that if neither Alice or Bob sees a repetition in their parts of the input string, then they should exchange all their letters to check for cross-part repetitions, which cannot be done in a constant number of messages -- 
    a detailed proof is given in \cref{sec:infinite-alphabets}. 
\end{myexample}

Which automaton model, if any, corresponds to protocols? As explained in the previous two examples, deterministic register automata are too weak, while nondeterministic register automata are too strong.
We conjecture (in Conjecture~\ref{conj:protocols-unambiguous}), that the winner is a seemingly  unexpected candidate, namely unambiguous register automata~\cite[Section 5]{colcombet2015unambiguity}. This is the special case of  nondeterministic register automata, in which for every input string there is at most one accepting run.
We discuss this conjecture in \cref{sec:infinite-alphabets}, and provide evidence in its favour. We start with an actual proof of one implication, namely:
\begin{align*}
\text{protocol} \quad \impliedby \quad \text{unambiguous automaton}.
\end{align*}
Contrary to  previous variants of this implication, the proof is non-trivial -- the usual construction does not work, because the automaton is nondeterministic. One  interesting phenomenon is that, in the case of infinite input alphabets, the interactive multi-round nature of the protocols becomes essential, and protocols cannot be reduced to one round, as was the case for finite input alphabets. In our proof of the implication $\impliedby$, we design a protocol where  the two parties  progressively eliminate the uncertainty about letters used in the run of the automaton, until the unique accepting run is identified or its existence disproved.
The proof uses a variant of the sunflower lemma.

Therefore,  the content of the conjecture is -- as in previous cases --  the other implication, namely that protocols can be simulated by unambiguous register automata. We provide evidence for this implication, using the recently developed theory of orbit-finite vector spaces~\cite{orbitFiniteVectorTheoretics}. We show that every function computed by a protocol can be computed by a weighted automaton with registers. This is almost like an unmabiguous automaton, except that some runs might have negative weights, and the weights always cancel out to give a final result that is either $0$ or $1$. In particular, the functions computed by protocols are computable, which is not a priori clear from the model. Along the way, we develop some new theory, in particular an orbit-finite generalisation of the Fliess Theorem.

\section{Boolean outputs}
\label{sec:boolean-domain}

In this section, we formally describe our model of computation for the simplest
output domain, namely the Booleans, and we prove that it defines exactly the
regular languages. (As mentioned in the introduction, this result was already shown in~\cite{hauser1989}.)  The definition that we describe, see
Definition~\ref{def:two-party-protocol-boolean} below, has minor differences
with respect to the informal description from the introduction. The messages
are not necessarily bits, but they belong to some finite set, which is fixed in
advance before the input is known. Also,  the two parties send their messages
in parallel in each round. These generalisations do not change the expressive
power of the protocol (they might influence the number of rounds), but they
will be useful in later sections, when we consider restrictions and
generalisations. 

\begin{definition}[Boolean protocol]
    \label{def:two-party-protocol-boolean}
   A \intro{Boolean protocol}
   is given by the following ingredients: 
  \begin{enumerate}
    \item a finite input alphabet $\Sigma$;
    \item a number of rounds $k \in \set{1,2,\ldots}$;
    \item message spaces for Alice and Bob, which are finite sets $Q_A$ and $Q_B$;
    \item for each round $i \in \set{1,\ldots,k}$, two strategies
    \begin{align*}
    \myoverbrace{\sigma_A^i : \myunderbrace{\Sigma^*}{Alice's \\ local\\ string} \times \myunderbrace{Q_B^{i-1}}{message \\ history}  \to \myunderbrace{Q_A}{new\\ message}}{stategy for Alice in the $i$-th round}
    \qquad \text{and} \qquad 
        \myoverbrace{\sigma_B^i : \myunderbrace{\Sigma^*}{Bob's \\ local\\ string} \times \myunderbrace{Q_A^{i-1}}{message \\ history}  \to \myunderbrace{Q_B}{new\\ message}}{stategy for Bob in the $i$-th round};
    \end{align*}
    \item an output function of type $(Q_A \times Q_B)^{k} \to \set{\text{yes, no}}$.
  \end{enumerate}
\end{definition}
\AP
Given an input string $w \in \Sigma^*$ and a split $w = w_1 w_2$ into two strings,  which are called the \emph{local
strings} of Alice and Bob, respectively, the  protocol is run as follows. There
are $k$ rounds. In each round, both parties send messages, and therefore after
$i$ rounds are played, the communication history contains $i$ messages sent by
Alice and $i$ messages sent by Bob.  In round  $i \in \set{1,\ldots,k}$,  each
of the two parties  looks at their local string and the $i-1$ messages sent by
the other party in the previous rounds (only the messages sent by the other
party are needed, since the party knows their own messages). Based on this
information, each party uses their strategy to produce a new message. At the
end of the protocol, the output function is used to determine the value of the
function, based on all messages in  the communication history. We say that the
protocol computes a language $L \subseteq \Sigma^*$ if for every string $w$ and every split $w = w_1 w_2$, the output of the protocol tells us if  $w$
belongs to the language. This corresponds to the split invariance condition
that was discussed in the introduction.  Also, the reader will recognise the
restriction on the total number of bits (this is bounded by the number of
rounds times the logarithm of the size of the message spaces), and the
non-uniformity (there is no restriction on the strategies of Alice and Bob). By
non-uniformity, the first message sent by Alice could contain an answer to some
undecidable problem. However, as we will see, the split invariance restriction
will make it impossible to use this information, since the protocol can only
compute regular languages, as stated in the \cref{thm:boolean-domain} from the
introduction, which we now recall.

\booleandomain*

\begin{proof}
\AP
  We begin with the easier right-to-left implication, which says that the protocol can compute every regular language. If the language is  regular, then it  is recognised by a deterministic finite automaton, say  with state space $Q$. To compute the language, we can use a
  \intro{one-round protocol} (as we will see in a moment, this is not a coincidence, since all protocols can be reduced to one round). Alice sends the state in $Q$ of the automaton after reading her local string, and Bob sends the dependency $Q \to \set{\text{yes, no}}$ which says how Alice's state determines acceptance. Once these two pieces of information are known, we can apply the function from Bob's message to the state in Alice's message to determine the output.

  The rest of this proof is devoted to the left-to-right implication, i.e.~to showing that every language computed by the protocol is regular.  We begin by reducing to one round. 
  \begin{lemma}\label{lem:one-round-reduction-boolean}
    For every protocol, there is an equivalent one-round protocol. 
  \end{lemma}
  \begin{proof}
    The general idea is that instead of engaging in interactive communication, each party sends the dependency of their message upon the unknown messages of the other party. Suppose that we are  Alice. The sequence of messages that we send will depend on our local string, and  the messages sent by Bob. Once the local string is fixed, this dependency is captured by a function of type 
    \begin{align*}
    (Q_B)^k \to (Q_A)^k,
    \end{align*}
    which satisfies the following \emph{causality} constraint: the $i$-th coordinate of the output depends only on the first $i-1$ input coordinates.
    Instead of waiting for Bob's messages, Alice sends this function. At the same time, Bob sends an analogous function of type 
    \begin{align*}
    (Q_A)^k \to (Q_B)^k,
    \end{align*}
    which describes the dependency of his messages. Due to the causality constraints, the two functions combine to create a unique output in $(Q_A \times Q_B)^k$, which can be used to determine the output of the protocol.
  \end{proof}

The proof of the above lemma incurs an exponential cost in the size of the message spaces. This is of little concern to us, since we only care about the protocol having a constant number of rounds.
  To complete the proof of the theorem, we use the  Myhill-Nerode Theorem.  Indeed, consider the strategy of  Alice: 
  \begin{align*}
  \sigma_A : \Sigma^* \to Q_A.
  \end{align*}
  For all we know, this function could be non-computable. However, it classifies all local strings into finitely many categories, with one category for each message in $Q_A$. Furthermore, if two strings $w_1$ and $w'_1$ are in the same category, then they are equivalent in the following sense: 
  \begin{align}\label{eq:myhill-nerode-equivalence}
  w_1 w_2 \in L \Leftrightarrow w'_1 w_2 \in L 
  \qquad \text{for every $w_2 \in \Sigma^*$.}
  \end{align}
  The equivalence described above is the same equivalence as in the Myhill-Nerode Theorem. In particular, equivalence classes of this equivalence are the same as states of the minimal deterministic automaton. Since the number of possible messages in $Q_A$ is finite, it follows that the minimal automaton is finite, and therefore the language is regular.  (Observe that we have shown that the language has at most $Q_A$ equivalence classes of Myhill-Nerode equivalence, which gives an upper bound on the size of the minimal automaton recognizing the language.)
\end{proof}

In \cref{thm:boolean-domain}, we have seen that protocols with a constant number of bits exchanged define exactly the regular languages.  Before continuing, let us briefly discuss what happens if a small, but nonconstant, number of bits is allowed (the observations in this discussion were suggested by Katzper Michno). Already by sending $\log n$ bits, where $n$ is the input length, one can compute any property of the length of the input string, including undecidable properties. This is because $\log n$ bits are enough to get the length of the input string. For one-round protocols, there is nothing between $\log n$ bits and a constant number. Indeed, as mentioned at the end of the proof of \cref{thm:boolean-domain}, the minimal number of bits that can be sent by Alice in a one-round protocol for a string of length $n$ is 
\begin{align*}
\log( \myunderbrace{\text{number of Myhill-Nerode equivalence classes for strings of length at most $n$}}{this is called the state complexity of the language} ).
\end{align*}
For a regular language the state complexity is constant. For a non-regular language, the state complexity cannot be smaller than $n$. This can be proved similarly to the Morse-Hedlund, see~\cite[Theorem 1.3]{ormes2016extender} for a generalisation which talks about multidmensional strings. We do not discuss this further, since our focus is on regular languages (and their generalisations to other output domains that will be dicussed in later sections), in which only a constant amount of communication is needed.
\section{Beyond Boolean outputs}
\label{sec:beyond-boolean-outputs}

In the previous section, we considered Boolean outputs. In this section, we
generalise protocols to account for an arbitrary \kl{output domain}. The inputs
remain unchanged -- they will always be strings in this paper.  For the
output domain, we use a very general notion, namely a set with some
operations. 
\begin{definition}[Output domain]
  \AP
    An \intro{output domain} consists of: 
    \begin{enumerate}
        \item an underlying set $\domain$; together with
        \item a family of \emph{operations}, each one having  type $\domain^n \to \domain$ for some $n \in \set{0,1,\ldots}$.
    \end{enumerate}
\end{definition}
An output domain is the same thing as a (non-indexed) algebra, in the sense of universal algebra~\cite[p.5]{hobby1988structure}. 
The output domain will typically  be infinite. In principle, the family of operations might be infinite as well,
although any protocol will only use finitely many operations.
By abuse of notation, we use the same symbol $\domain$ to denote the output domain and its underlying set, with the operations being implicit. Here are the output domains that will be studied in this paper: 
\begin{itemize}
    \item \intro{Boolean domain}. The set is $\set{\text{true, false}}$, and there are no basic operations. (We could add basic operations, such as the Boolean operations $\lor,\land$ and $\neg$, but this will not affect the expressive power of our protocol, so we choose to have no operations.)
    \item \intro{Field domain}. The set is a field, such as the rationals or reals, and there are two operations for addition and multiplication. This is not one domain, but a family of domains, with one for each field. Division is not included as an operation.
    \item \intro{String domain}. The set is $\Gamma^*$ for some finite alphabet $\Gamma$, and the operation is string concatenation.
\end{itemize}

In the protocol, the output value will be constructed in a constant number of
steps, by using operations from the output domain. We will not distinguish
between the operations that are in the output domain, and other operations
that are derived by composing them, as described in the following definition.

\begin{definition}[Term operation]\label{def:term-operations}
    Consider an output domain $\domain$. An operation 
    \begin{align*}
    t : \domain^n \to \domain
    \end{align*}
    is called a \intro{term operation} if it can be obtained by applying the operations in the domain to variables $x_1,\ldots,x_n$. Each variable can be used multiple times, or not at all.  We write 
    \begin{align*}
    \domain^n \termop \domain
    \end{align*}
    for the set of term operations with $n$ arguments.
\end{definition}

\begin{myexample}\label{ex:identity-term}
    Even if the output domain has no operations, the identity operation $\domain \to \domain$ is allowed, since it is defined by a term that consists of just a variable $x_1$.
\end{myexample}

\begin{myexample}
  If the operations are $+$ and $\times$, then after applying the usual distributivity laws,  \kl{term operations} are the same as polynomials with natural coefficients, such as 
\begin{align*}
x_1x_2^2 + x_1x_2x_3^3  + \myunderbrace{2x_1}{same as $x_1 + x_1$}.
\end{align*}
If the output domain is the \kl{Boolean domain}, which has no operations, then the only kind of \kl{term operation} is a single variable, e.g.~$x_2$.  In the case of a \kl{string domain}, a \kl{term operation} is some concatenation of the variables, such as $x_1 x_3 x_1 x_2$.
\end{myexample}

\AP
We now generalise the \kl[boolean protocol]{protocol} to cover functions 
\begin{align*}
f : \Sigma^* \to \domain,
\end{align*}
where $\domain$ is some possibly infinite output domain. As in the Boolean
version of the protocol, the output value is constructed by Alice and Bob, as a
result of an exchange of a constant number of messages. Each message has two
parts, which are called the \intro{signal part} and the \intro{output part}. The
signal part consists of a finite amount of information, and is used to exchange
information between the two parties as in the \kl{Boolean protocol}. The \kl{output part}
is a list of elements from the  output domain, and  is meant to be part of the
output value. As in the Boolean case, the above definition slightly deviates
from the informal description in the introduction. In particular, a message can
contain $d$ elements of the output domain. This does not affect the expressive
power of the protocol, but it might affect the number of rounds, and the above
definition will be more convenient later one, where we consider one-round
protocols. Another important difference is that the operations from the output
domain are only applied once at the end of the protocol. This, again does
affect the expressive power, since the two parties can wait with their
operations until all \kl[signal part]{signals} have been exchanged.  

\begin{definition}\label{def:two-party-protocol-general}
  A \intro{two-party protocol}
   is given by the following ingredients: 
  \begin{enumerate}
    \item an \kl{output domain} $\domain$;
    \item a finite input alphabet $\Sigma$;
    \item a \intro{number of rounds} $k \in \set{1,2,\ldots}$;
    \item \kl{signal spaces} for Alice and Bob, which are finite sets $Q_A$ and $Q_B$;
    \item a \intro[protocol dimension]{dimension} $d \in \set{0,1,\ldots}$;
    \item for each round $i \in \set{1,\ldots,k}$, a \intro{strategy}
    \begin{align*}
    \myoverbrace{\sigma_A^i : \myunderbrace{\Sigma^*}{Alice's \\ local\\ string} \times \myunderbrace{Q_B^{i-1}}{history of \\ signals \\ from Bob}  \to \myunderbrace{Q_A \times \domain^d}{new\\ message}}{stategy for Alice in the $i$-th round}
    \qquad 
        \myoverbrace{\sigma_B^i : \myunderbrace{\Sigma^*}{Bob's \\ local\\ string} \times \myunderbrace{Q_A^{i-1}}{history of \\ signals \\ from Alice}  \to \myunderbrace{Q_B \times \domain^d}{new\\ message}}{stategy for Bob in the $i$-th round}
    \end{align*}
    \item an \intro[protocol output function]{output function} of type \begin{align*}
    (Q_A \times Q_B)^{k} \to (\domain^{2dk} \termop \domain).
    \end{align*}
  \end{enumerate}
\end{definition}

\AP
The \kl{protocol} is executed on a pair of strings $(w_1,w_2)$. In each round, each
of the two parties  sends a message (which consists of a signal and some
elements of the output domain) that is based on their local string and the
history of signals coming from the other party. After all $k$ rounds have been
executed, the joint signal history of both players is used, by the output
function, to determine a term operation. This operation is then applied to the
output history, yielding the final result of the protocol. As in the Boolean
case, we are interested in protocols that are \intro{split invariant}, which means that
for every string $w \in \Sigma^*$, the same output is produced for every
possible decomposition $w = w_1 w_2$. Such protocols compute a function of type
$\Sigma^* \to \domain$. 

\begin{myexample}[Boolean domain]
    Consider the Boolean output domain. In this case, the distinction between
    \kl{output values} and \kl{signals} is irrelevant, since the output values
    can be sent as signals. To produce the output, one of the parties sends it in a message, and the output function is a single variable $x_i$ that copies it, as in Example~\ref{ex:identity-term}. Therefore, the general protocol coincides with the
    \kl{Boolean protocol} from the previous section, and  can only compute regular
    languages.  The same remarks apply for a general but finite output domain
    -- a function $f : \Sigma^* \to \domain$ can be computed by a \kl{protocol} if
    and only if for every $d \in \domain$, the inverse image $f^{-1}(d)$ is a
    regular language.
\end{myexample}

Other examples of output domains are fields and strings. These were mentioned in the introduction in Sections~\ref{sec:intro-field}
and~\ref{sec:intro-strings},  and will be discussed at length in
Sections~\ref{sec:field-domain} and~\ref{sec:string-outputs}.

We have little say to say about \kl{protocols} in the case of a general output
domain. The  only result that we have at this level of generality is a
reduction to one-round protocols, similarly to the \kl[Boolean
protocol]{Boolean case}.

\begin{lemma}
\label{lemma:one-round-reduction-general}
  Every \kl{protocol} is equivalent to a one-round \kl{protocol}.
\end{lemma}
\begin{proof}
  The proof of \cref{lem:one-round-reduction-boolean} generalises to 
  arbitrary output domains with minor modifications. Here we need to take care of elements of the output domain that are sent in messages. This kind of dependecy similarly captured by functions of type:
  \begin{align*}
    (Q_B)^k \to \domain^{dk}
  \end{align*}
  for Bob, and
  \begin{align*}
    (Q_A)^k \to \domain^{dk}
  \end{align*}
  for Alice. These functions can be included in the first message that each party sends, and therefore the same construction as in the proof of \cref{lem:one-round-reduction-boolean} applies.
\end{proof}

Slightly ahead of time, we mention that the reduction to one-round protocols will no longer be valid for infinite input alphabets, which will be studied in Section~\ref{sec:infinite-alphabets}. As we will see, in that case  the interactive nature of the protocols
will be essential.

\section{Field outputs}
\label{sec:field-domain}
In this section, we discuss functions where the output domain is a field, equipped with addition and multiplication. We prove that protocols have  exactly the same expressive power as weighted automata.
We begin by recalling the notion of weighted automata.

\paragraph*{Weighted automata.} \AP A \kl{weighted automaton} is a device that
is used to compute a function from strings to a field (more generally, a
semiring, but we consider the case of fields here). This model was originally
introduced by \schutz~\cite{schutzenberger1961definition}. A weighted automaton is defined like a nondeterministic automaton, except that the transitions have weights in the field, and instead of choosing subsets of initial and final states, we also have assignments of weights. A weighted automaton makes sense not only for fields, but also for semirings, so we define it that way. 
\begin{definition}
    [Weighted automaton] 
    \label{def:weighted-automaton-nondeterministic}
    \AP
    A \intro[wa-nfa]{weighted automaton} over a semiring $\domain$ consists of a finite input alphabet $\Sigma$,  a finite set of states $Q$, and functions: 
    \begin{align}
        \label{eq:weight-functions}
    \myunderbrace{I : Q \to \domain}{initial}
    \quad
    \myunderbrace{F : Q \to \domain}{final}
    \quad
    \myunderbrace{\Delta : Q \times \Sigma \times Q \to \domain}{transitions}.
    \end{align}
\end{definition}
A run of this automaton is defined in the usual way: it is a sequence of
transitions, one for each input letter, such that consecutive transitions
agree on the connecting states. The weight of a run is the product of: (1) the
initial weight of its source state; (2) the weights used by its transitions;
and (3) the final weight of its target state. If the semiring is
non-commutative, the order of multiplication is  important, and transitions
are multiplied in the order corresponding to the input string. For an input
string, the output of the automaton is the sum of weights of all runs over
this automaton (sum is always commutative).

The main result of this section is that our protocol is equivalent to weighted automata, as stated in the following theorem from the introduction, which we now recall:
\fielddomain*


As we will see in Example~\ref{ex:non-commutative-semirings} below, the theorem is false for general semirings that are not fields. It is possible that the theorem extends to commutative semirings; there is some evidence supporting this conjecture (for example, the supports of the relevant functions appear to be regular), but it remains unproven. The problem with such generalisations is that our proof uses the  Fliess Theorem, which is only known  for fields.  Before proving the theorem in
Section~\ref{sec:proof-of-thm-field-domain}, we return to the issue of
division, which was already discussed in Example~\ref{ex:division}.

\begin{myexample}[Division, continued]\label{ex:division-continued}
    Because it is undefined for zero, division is not a total operation, and
    therefore technically speaking it does not fall into our framework. We
    could, however try to incorporate it, by making the two parties responsible
    for avoiding division by zero. Under this framework, we could use a
    \kl{protocol} to compute the function $1/|w|$ (a better choice would be
    $1/(|w|+1)$, since it would avoid problems with the empty string). As we have
    discussed in Example~\ref{ex:division}, such a function cannot be computed
    by a \kl{protocol} that uses only addition and multiplication. We do not know
    what functions can be computed if division is also allowed.
\end{myexample}

\begin{myexample}[Semiring outputs]
    \label{ex:non-commutative-semirings} 
    In this example we show that for semirings which are not fields, the \kl{protocol} need not be equivalent to \kl{weighted automata}. The implication 
    \begin{align*}
    \text{protocol} \quad \impliedby \quad \text{weighted automaton}
    \end{align*}
    in \cref{thm:field-domain}, as we will see in a moment, holds for any
    semiring, and therefore the problematic implication is the other one. Here
    is an example where it fails. Let $\domain$ be  the free (non-commutative)
    idempotent semiring generated by two letters $a$ and $b$. Elements of this
    semiring are finite sets of words in $\set{a,b}^*$, such as 
    \begin{align*}
    \set{3ab, 5ba, 7aab}
    \end{align*}
    The addition operation is multiset union, 
    and the multiplication operation is concatenation of words, 
    extended to sets in the natural way, as illustrated on this example:
    \begin{align*}
    \set{a,b}\cdot \set{a,b} = \set{aa, ab, ba, bb}.
    \end{align*}
    \kl{Weighted automata} over this semiring are the same as the 
    rational relations~\cite[Chapter IX]{Eilenberg74}. 
    On the other hand, a protocol can define string-to-$\domain$ functions
    that are not rational. This is witnessed already by functions that produce 
    singleton sets (call these singleton functions), which can be seen as 
    functions of type $\Sigma^* \to \set{a,b}^*$. For example, consider 
    the singleton version of the  reverse function, i.e.
    \begin{align*}
    w \mapsto \set{\text{reverse of $w$}} \in \domain.
    \end{align*}
    This function can be computed by a \kl{protocol}, using the same idea as in 
    Example~\ref{ex:reverse-duplicate}. This function, however, is not
    a rational relation, and therefore it is not computed by a weighted automaton over $\domain$. 

    For protocols over this semiring where the outputs are always singletons, we can connect with the results on string-to-string functions that will be discussed in Section~\ref{sec:string-outputs}.
In this case, all messages sent during the protocol must be singletons (this is because once a non-singleton is produced, it can never be turned into a singleton). Therefore, the operation $+$ can never be used in a non-trivial way, and thus the protocol can only use multiplication. This means that it coincides with the protocols with outputs that are strings with concatenation, as discussed in Section~\ref{sec:string-outputs}. According to Conjecture~\ref{conj:protocol-regular-string-to-string}, the singleton functions are therefore exactly the \kl{regular functions}. It is worth mentioning that, in view of \cref{thm:field-domain}, two‑way weighted automata over fields do not have greater expressive power than one‑way weighted automata: every two‑way weighted automaton over a field can be simulated by a protocol over the same field, and again by \cref{thm:field-domain} that protocol can be simulated by a one‑way weighted automaton. The example above shows that this equivalence can fail for general semirings. But similar to one direction of the proof of \cref{thm:field-domain}, protocols over a general semiring can simulate two‑way weighted automata over the same semiring, but the converse is not clear for us. We conjecture that, for general semirings, protocols correspond to two‑way weighted automata; however, we currently have no evidence for.

\end{myexample}

\begin{myexample}[Equality tests]
\label{ex:equality-tests}
In this example, we discuss an extension of the \kl{protocol} which allows for
equality tests, similarly to the algebraic group
model~\cite{fuchsbauer2018algebraic}. Clearly, equality tests cannot be
completely unrestricted. Otherwise, in the presence of a countable output
domain (which is the case for all protocols studied in this paper), the
receiver could compare the message with all possible values one by one, until
the correct one would be identified. This would  invalidate the \kl{black box
discipline}. A reasonable restriction is to allow a constant number of equality
tests for each message; this constant can also be brought down to one, by
possibly sending more copies of the same message. The resulting protocol would
be able of complementing a weighted automaton $\Aa$, in the following sense:
\begin{align*}
w \in \Sigma^* 
\quad \mapsto \quad 
\begin{cases}
    1 & \text{if $\Aa(w) =0$}\\
    0 & \text{otherwise}.
\end{cases}
\end{align*}
This form of complementation is undesirable from the point of view of decidability. For example, language equivalence is undecidable for \kl{weighted automata} that are complemented in this way~\cite[Theorem 4.9]{bojanczyk_automata_2025}. 
Since we strive for protocols that describe ``regular'' functions, and such functions should be decidable, we avoid equality tests.
\end{myexample}

\begin{myexample}[Wrong output domains]\label{ex:wrong-output-domains}
 This discussion of equality tests from Example~\ref{ex:equality-tests} also
 explains why we should not expect results about regularity that work for any
 \kl{output domain}. For example, if we would extend the field domain with a unary
 complementation operation 
 \begin{align*}
 x 
 \quad \mapsto \quad 
 \begin{cases}
    1 & \text{if $x =0$}\\
    0 & \text{otherwise},
\end{cases}
 \end{align*}
 then our protocols could recover the undecidable model discussed in the
 previous paragraph.  Of course, one can come up with even more obviously wrong
 \kl{output domains}, such as a domain that consists of Turing machines with certain
 evaluation operations. We do not know where the dividing line is between
 ``right'' and ``wrong'' \kl{output domains}.
\end{myexample}

\subsection{Proof of \cref{thm:field-domain}}
\label{sec:proof-of-thm-field-domain}
\AP
We now return to the proof of \cref{thm:field-domain}. The right-to-left
implication says that every \kl{weighted automaton} can be simulated by a \kl{protocol}.
This is proved essentially in the same way as in the \kl[Boolean protocol]{Boolean case}, and the proof is valid for all semirings and not just fields. 
Suppose that
the function is computed by a \kl{weighted automaton}, which has state space $Q$. To every input string $w \in \Sigma^*$, we can associate a $Q \times Q$ matrix over the semiring. This matrix is defined by 
\begin{align*}
M[p,q] = \sum_\rho \text{product of weights of transitions used by $\rho$},
\end{align*}
where $\rho$ ranges over runs of the weighte automaton that start in state $p$, read the input string $w$, and end in state $q$. This map is a homomorphism from strings to matrices, i.e.~the matrix corresponding to a string $w_1 w_2$ is obtained by multiplying the matrices corresponding to $w_1$ and $w_2$. In the \kl{protocol}, Alice sends
the matrix  which corresponds to her local string, and Bob sends the  matrix
which corresponds to  his local string. These matrices are multiplied using the
semiring operations, and then multiplied with vectors that represent the initial and final weights of states. This
protocol has one round and is \kl{signal-free}, i.e.~no information is conveyed
using signals.

The rest of this proof is devoted to the left-to-right implication,
i.e.~showing that every function computed by a \kl{protocol} is computed by a
\kl{weighted automaton}. As in the Boolean case, we will do a sequence of
reductions, such that the \kl{protocol} becomes more and more restrictive. In
particular, we will show that the \kl{protocol} can be reduced to a version that has
one-round and is \kl{signal-free}.

\subsubsection{Reduction to a scalar product protocol}
\label{sec:reduction-to-scalar-product-protocols}
\AP
In the first step, we show that each \kl{protocol} can be constrained to have a
special form, which has one round and is \kl{signal-free}. This protocol uses only
the scalar product,  as explained in the following definition. 

\begin{definition}[Scalar product protocol] \label{def:scalar-product-protocol}
    Assume that the \kl{output domain} is a field.
    A \intro{scalar product protocol} is defined as follows.
    First, each of the two parties uses their local string to 
    produce a vector of field elements, 
    of some fixed dimension $d$, as expressed by two functions: 
    \begin{align*}
    \sigma_A, \sigma_B : \Sigma^* \to \domain^d.
    \end{align*}
    Next, the output is defined to be the scalar product of the two vectors. 
\end{definition}

This \kl[scalar product protocol]{protocol} has the same power as \kl{general
protocols}. Furthermore, the proof works not only for fields, but also for the more general case of commutative semirings (these are semirings where both addition and muliplication are commutative, in contrast to general semirings, where only addition is assumed to be commutative). The assumption that the output domain is a field, and not just a commutative semiring, will be used at a later stage, when we apply the Fliess Theorem.

\begin{lemma}\label{lem:scalar-product-reduction}
  Assume that the \kl{output domain} is a commutative semiring. 
  If a function is computed by a \kl{protocol}, 
  then it is computed by a \kl{scalar product protocol}.
\end{lemma}
\begin{proof}
    The proof is a sequence of reductions, 
    where more and more conditions are imposed on the protocol.  
    
    \paragraph*{Step 1. One-round protocol.} 
    The first step is to reduce the protocol to a one-round protocol. 
    This is done using \cref{lemma:one-round-reduction-general}.

 \paragraph*{Step 2. Signal-free protocol.}  \AP 
 We say that a protocol is
 \intro{signal-free} if both of the sets $Q_A$ and $Q_B$ have one element each.
 In other words, the signals do not convey any information, and the only
 messaging activity consists of sending elements of the \kl{output domain}. In a
 \kl{signal-free protocol}, the concept of rounds is irrelevant, since the behaviour
 of one party is not influenced by the communication from the other party. The following claim shows that every protocol can be made signal-free, even if the semiring is not necessarily commutative. 


 \begin{claim}
    \label{claim:trivial-messages}
    Assume that the \kl{output domain} is a semiring, not necessarily commutative. 
    Then every one-round protocol is equivalent to a \kl{signal-free protocol}.
 \end{claim}
 \begin{proof} 
    Consider a one-round protocol. Without loss of generality, we assume that
    both signal spaces $Q_A$ and $Q_B$ are the same space $Q$. (We can always
    use the union of two signal spaces for both parties.) Assume that each of
    the parties sends $d$ semiring elements in the protocol. In other words, the
    protocol works as follows:
    \begin{enumerate}
        \item Based on her local string, Alice chooses a message $(q_A,\bar x) \in Q \times \domain^d$;
        \item Based on his local string, Bob chooses a message $(q_B,\bar y) \in Q \times \domain^d$;
        \item Based on the signals $q_A$ and $q_B$, a term operation  with $2d$ variables is chosen, call it $t_{q_A,q_B}$, and the output is obtained by applying this term operation to $(\bar x, \bar y).$
    \end{enumerate}
    To prove the claim, we need to show that the protocol can be adapted so
    that always the same term operation is chosen, i.e.~there is no dependence
    of this term operation on the signals $q_A$ and $q_B$. This way the signals
    can be eliminated. To do this, we increase the \kl{protocol dimension} from $d$ to $d +
    |Q|$. 

    This means that for each possible signal $q \in Q$, each party sends a
    semiring element corresponding to this signal. This element will be either $0$ or $1$, which are elements that need to exist in every semiring (see~\cite[p.7]{handbook2009} for a formal definition of semirings). The idea is that instead of
    sending a signal $q \in Q$, each party will set the corresponding semiring
    element to $1$, and the remaining semiring elements to $0$. The corresponding
    \kl{term operation} is then 
    \begin{align*}
    \sum_{\substack{q_A \in Q \\ q_B \in Q}} \myoverbrace{x_{q_A} \cdot y_{q_B}}{variables corresponding \\ to the messages $q_A$ and $q_B$, } \cdot t_{q_A,q_B}(\bar x, \bar y).
    \end{align*}
    When evaluating this \kl{term operation},
    the summands that do not correspond to the intended message $(q_A,q_B)$
    will be eliminated, since they will contain a variable that is set to $0$. 
    Only the summand corresponding to the intended message will be used,
    and thus the correct output will be produced. 
 \end{proof}

 \paragraph*{Step 3. Scalar product.}
 In the previous step, we have reduced the \kl{protocol} to a special case, where
 Alice and Bob send vectors, call them $\bar x, \bar y \in \domain^d$, and then
 some fixed \kl{term operation} $t$ with $2d$ variables is applied to them. To
 complete the proof of the lemma, we show that the term operation can be turned
 into a scalar product. This term operation is a sum of monomials, with each
 monomial being a product of some variables. This part of the proof will work for commutative semirings (the previous step worked for all semirings). 
 
 Consider the monomials in the term
 operation $t$. Since muliplication is assumed to be  commutative, for each monomial, its contribution to the output  is obtained
 by multiplying two numbers: (a) the product of the  variables in the term
 operation that are contributed by Alice; and   (b) the product of the
 variables in the term operation that are contributed by Bob. We can redesign
 the protocol so that for each monomial, Alice sends the contribution (a), and
 Bob sends the contribution (b). In the new protocol, the dimension is the
 number of monomials from the original protocol, and the term operation is a
 scalar product. 
\end{proof}

\subsubsection{From a scalar product protocol to a weighted automaton}
\label{sec:from-scalar-product-protocol-to-weighted-automaton}
\AP
In this section, we complete the proof of \cref{thm:field-domain}, by showing
that \kl{scalar product protocols} can be simulated by \kl{weighted automata}.
Similarly to the \kl[boolean protocol]{Boolean case}, the proof uses a
Myhill-Nerode characterization. In the case of \kl{weighted automata}, this
characterization is called  the Fliess Theorem, which characterizes
functions computed by weighted automata in terms of a certain infinite matrix. It is here where we use the assumption that the output domain is a field, and not just a commutative semiring.

\begin{definition}[Hankel Matrix]\label{def:hankel-matrix}
  \AP
    Let $\domain$ be a field. The \intro{Hankel matrix} of a function 
    \begin{align*}
    f : \Sigma^* \to \domain
    \end{align*}  
    is the matrix where rows are words in $\Sigma^*$, columns are words in $\Sigma^*$, and the entry corresponding to a row $u$ and a column $v$ is $f(uv)$.
\end{definition}

Another perspective on the \kl{Hankel matrix} is that it describes the
\intro{derivatives} of the function $f$. Each row in the \kl{Hankel matrix} can be
seen as a function of type $\Sigma^* \to \domain$, which inputs columns
(i.e.~strings) and outputs the corresponding entries in the Hankel matrix. If
the row corresponds to a word $w$, then this function is
\begin{align*}
v \mapsto f(wv),
\end{align*}
which is called the \intro{left derivative} of $f$ with respect to $w$.
Similarly, the columns of the Hankel matrix describe \intro{right derivatives} of $f$.

The Fliess Theorem~\cite[Theorem 2.1.1]{fliess1974} states that a function 
\begin{align*}
f : \Sigma^* \to \domain
\end{align*}
is computed by a \kl{weighted automaton} if and only if  its \kl{Hankel matrix}
has finite rank, i.e.~its rows (i.e.~the \kl{left derivatives}) are spanned by
a finite subset. (This is equivalent to saying that the columns, or \kl{right
derivatives}, have a finite spanning subset.) Therefore, to complete the proof
of \cref{thm:field-domain}, it is enough to show the following lemma.

\begin{lemma}\label{lem:hankel-finite-rank}
    If a function is computed by a \kl{scalar product protocol},
    then its \kl{Hankel matrix} has finite rank.
\end{lemma}
\begin{proof}
  Essentially by definition, the \kl{Hankel matrix} of a function computed by a
  \kl{scalar product protocol} with dimension $d$ can be obtained as a
  sub-matrix of the following matrix: rows and columns are vectors in
  $\domain^d$, and the entries are obtained by taking scalar products. This
  matrix is easily seen to have finite rank, namely $d$, since the scalar
  product  becomes a linear operation once one of the two arguments is fixed.
\end{proof}

\section{String outputs}
\label{sec:string-outputs}
\AP
In this section, we consider the case where the \kl{output domain} is strings over
some finite alphabet. We use the name \intro{string-to-string function} for any
function of type $\Sigma^* \to \Gamma^*$, where both alphabets $\Sigma$ and
$\Gamma$ are finite. For such functions, the \kl{protocols} are assumed to use the
output domain of strings $\Gamma^*$ equipped with concatenation. In the case of
\kl{string-to-string functions}, we conjectured, see
Conjecture~\ref{conj:protocol-regular-string-to-string},  that protocols define
exactly the so-called \kl{regular functions}, which will be formally defined in
Section~\ref{sec:regular-string-to-string-functions} below.
 

\newcommand{\functionclass}[2]{
    \text{#1} & 
\left\{{\text{\begin{minipage}{11cm}
    \small 
#2
\end{minipage}}}\right. 
}

\newcommand{\vertinclusion}[1]{\\ &
\mathrel{\rotatebox[origin=c]{270}{$\subsetneq$}}
\text{\qquad \small(#1)}\\ \\ }

\begin{figure}
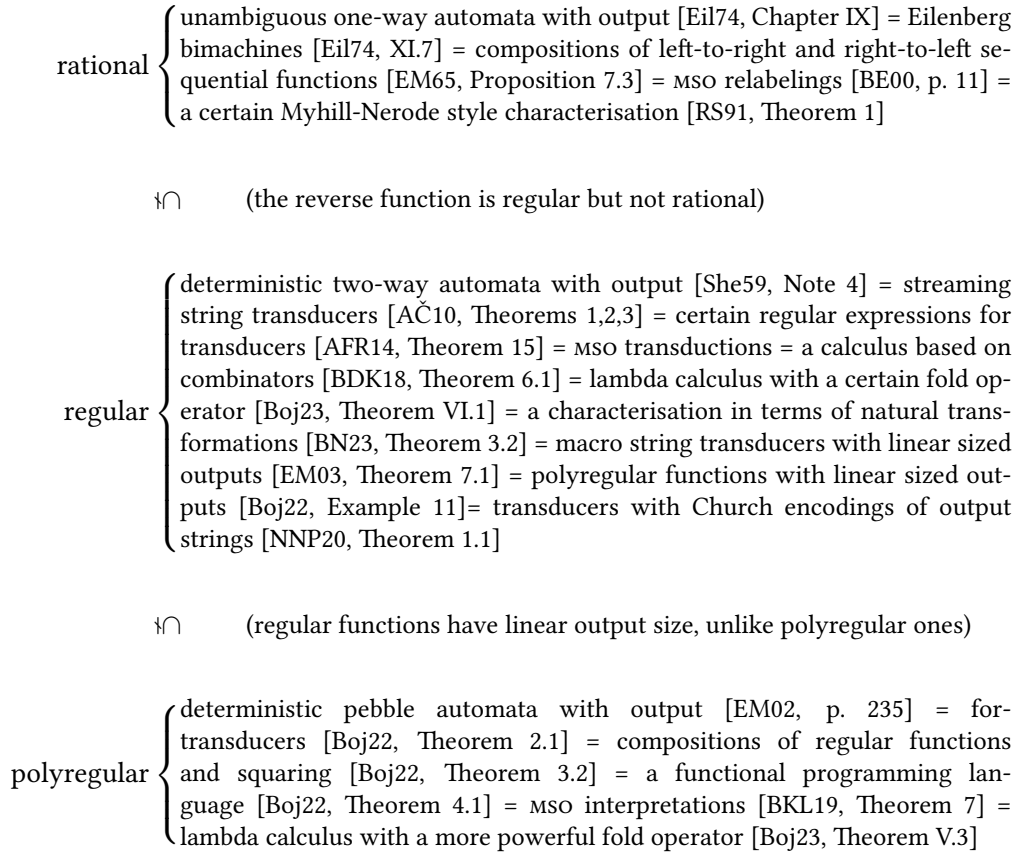

    \centering
\begin{align*}
    \functionclass{rational}{
        unambiguous one-way automata with output~\cite[Chapter IX]{Eilenberg74}  =  Eilenberg bimachines~\cite[XI.7]{Eilenberg74} =  compositions of left-to-right and right-to-left sequential functions~\cite[Proposition 7.3]{elgot_relations_1965}   = \mso relabelings~\cite[p.~11]{bloem_comparison_2000} = a certain Myhill-Nerode style characterisation~\cite[Theorem 1]{reutenauerSchutzenberger1991}
    } \\
\vertinclusion{the reverse function is regular but not rational}
\functionclass{regular}{ deterministic two-way automata with output~\cite[Note 4]{shepherdson1959reduction} = streaming string transducers~\cite[Theorems 1,2,3]{alurExpressivenessStreamingString2010} = certain regular expressions for transducers~\cite[Theorem 15]{alur2014regular}  = \mso transductions = a calculus based on combinators~\cite[Theorem 6.1]{bojanczykRegularFirstOrderList2018} = lambda calculus with a certain fold operator~\cite[Theorem VI.1]{polyregular-fold} =   a characterisation in terms of natural transformations~\cite[Theorem 3.2]{bojanczykTitoRegular23} = macro string transducers with linear sized outputs~\cite[Theorem 7.1]{engelfrietMacroTreeTranslations2003} = polyregular functions with linear sized outputs~\cite[Example 11]{polyregular-survey}= transducers with Church encodings of output strings~\cite[Theorem 1.1]{implicit2}  } \\
\vertinclusion{regular functions have linear output size, unlike polyregular ones}
\functionclass{polyregular}{
deterministic pebble automata with output~\cite[p.~235]{engelfriet2002two} = for-transducers~\cite[Theorem 2.1]{polyregular-survey} = compositions of regular functions and squaring~\cite[Theorem 3.2]{polyregular-survey} = a functional programming language~\cite[Theorem 4.1]{polyregular-survey} = \mso interpretations~\cite[Theorem 7]{msoInterpretations} = lambda calculus with a more powerful fold operator~\cite[Theorem V.3]{polyregular-fold} 
} 
\end{align*}
    \caption{Three important classes of string-to-string functions.}
    \label{fig:transducer-classes}
\end{figure}

The content of this section is an extended discussion of this conjecture.
In Section~\ref{sec:regular-string-to-string-functions} we  formally
define the \kl{regular functions} and  prove the implication 
\begin{align*}
\text{protocol} \impliedby \text{regular}.
\end{align*}
The content of the conjecture is therefore the implication 
\begin{align*}
\text{protocol} \implies \text{regular}.
\end{align*}

In Section~\ref{sec:continuity}, we present some evidence for this implication.
We show that \kl{string-to-string functions} computed by \kl{protocols} share
many good properties of the \kl{regular functions}, such as linear output size
and computability. In Section~\ref{sec:unary-output-alphabet}, we present
further evidence for the implication, namely we  prove it in the special case
where the output alphabet has only one letter (the remaining case is two output
letters, since more letters do not change the situation). 

\subsection{Regular string-to-string functions}
\label{sec:regular-string-to-string-functions}
\AP
In this section, we define the class of \kl{regular string-to-string functions}, and
we prove the implication $\impliedby$ in the conjecture. Historically, class of
regular string-to-string  functions was first defined in terms of \kl{deterministic
two-way automata with output}~\cite[Note 4]{shepherdson1959reduction}. Although
numerous equivalent definitions appeared later on, we will use the original
definition.

 \begin{definition}[Two-way automaton]
   \AP
   A \intro{deterministic two-way automaton with output} is
   given by the following ingredients:
    \begin{enumerate}
        \item a finite input alphabet $\Sigma$;
        \item a finite output alphabet $\Gamma$;
        \item a finite set of states $Q$, with an initial state $q_0 \in Q$;
        \item a transition function  
        \begin{align*}
        \delta : 
        \myunderbrace{Q}{old \\ state} \times 
        \myunderbrace{(\Sigma + \{\vdash, \dashv\})}{input letter\\ under  the head} \to  \set{\text{halt}} + (
        \myunderbrace{Q}{new \\ state}
         \times 
         \myunderbrace{\{-1,0,1\}}{head \\ movement} \times 
         \myunderbrace{\Gamma^*}{added \\ output}) .
        \end{align*}
    \end{enumerate}
 \end{definition}

\AP
The automaton works as follows. The input string $w$ is placed on a tape,
with the left end marked by $\vdash$ and the right end marked by $\dashv$.
The automaton starts in state $q_0$, with its head on the left end of the
tape, which contains the marker $\vdash$. In each step, the automaton looks
at its current state and the letter under its head, and based on this
information, it uses the transition function to decide if it halts, or it
continues its computation. In case it continues, it chooses a  new state,
the direction in which it moves its head, and a string over the output
alphabet $\Gamma$, which is appended to the output tape. We assume that the
automaton is always halting, which means that for every input string, the
computation eventually halts. In particular, the computation must be
well-defined, which means that the head never falls off the input by moving
outside the endmarkers.   The semantics of such an automaton is of type
$\Sigma^* \to \Gamma^*$. (For automata which are not necessarily halting,
the function would be partial, since it would be undefined for inputs where
the automaton does not halt.)

    \begin{myexample}[Reverse]
        For each input alphabet $\Sigma$, the reverse function of type
        $\Sigma^* \to \Sigma^*$ is computed by a \kl{two-way automaton}, which
        first moves its head to the end of the string, and then starts copying
        it to the output while moving in the left direction.
    \end{myexample}

    The class of functions computed by \kl{two-way automata} has a remarkable number
    of equivalent descriptions, originating in different fields, including:
    monadic second-order transductions~\cite[Section
    4]{engelfrietMSODefinableString2001}, streaming string
    transducers~\cite[Section 3]{alurExpressivenessStreamingString2010},
    certain kinds of regular expressions~\cite[Section 2]{alur2014regular}, a
    calculus of functions based on  combinators~\cite[Theorem
    6.1]{bojanczykRegularFirstOrderList2018}, a characterisation based on
    natural transformations~\cite[Theorem 3.2]{bojanczykTitoRegular23}. For
    this reason, some authors (starting with Engelfriet and Hoogeboom), use the
    name \emph{regular} for this class of function, with the intended meaning
    being that these functions are the functional analogue of regular
    languages. Although this thesis can be questioned,
    see~\cite{polyregular-survey}, for the purposes of this paper we adopt the
    terminology of Engelfriet and
    Hoogeboom~\cite[p.~217]{engelfrietMSODefinableString2001}, as expressed in
    the following definition.

    \begin{definition}[Regular string-to-string function]
        \label{def:regular-string-to-string}
        A \kl{string-to-string function} is called \intro[regular function]{function}
        if it is computed by a \kl{deterministic two-way automaton with output}.
    \end{definition}

    One good  property of the \kl{regular string-to-string functions} is that
    they are closed under composition~\cite[Theorem
    2]{chytilSerialComposition2Way1977}. In particular,
    Conjecture~\ref{conj:protocol-regular-string-to-string} would imply that
    the same is true for functions computed by \kl{protocols}. Without proving
    the conjecture, we do not see any direct way of proving composition for
    protocols. The  \kl{regular string-to-string functions} have other good
    properties, such as decidable equivalence~\cite[Theorem
    1]{gurariEquivalenceProblemDeterministic1982}, but we do not discuss these
    in more detail, since they seem to have little direct bearing on protocols.

    The following lemma shows one of the implications in the conjecture.

\begin{lemma}\label{lem:from-regular-to-protocol}
    If a \kl{string-to-string function} is \kl[regular function]{function}, 
    then it is computed by a \kl{protocol}.
\end{lemma}
\begin{proof}
    The two parties can simulate a \kl{two-way automaton with output}. The execution
    of the protocol describes the crossing sequence of the automaton, i.e.~how
    it crosses the boundary between the two local strings of Alice and Bob.
    Here is a picture: \mypic{1} More formally, the crossing sequence is
    defined as follows, given  a split of the input string into two parts $w_1
    w_2$. We run the automaton until the first configuration which is in the
    word $w_2$. Then we run it until the first configuration which is in the
    word $w_1$. We continue this way, with odd-numbered steps describing runs
    inside $w_1$ that end in  configurations from $w_2$, and even-numbered
    steps describing runs inside $w_2$ that end in a configuration from $w_1$.
    The last step is exceptional, since it ends with an accepting
    configuration. The number of steps in a crossing sequence is bounded by twice the number of states, since otherwise the automaton would enter an infinite loop. This bound is the number of rounds in the protocol. In each round,
    the state corresponding to this round is sent as a signal, and the output
    value in the message is the  part of the output string that is produced in
    this  step. At the end of the protocol, the pieces of the output string are
    concatenated. 
\end{proof}

In view of the above lemma, the content of the conjecture is the opposite
implication, namely that every protocol computes are regular function.  The
rest of Section~\ref{sec:string-outputs} is devoted to  evidence for the
opposite implication.

\subsection{Evidence for the conjecture}
\label{sec:continuity}
\AP
In this subsection, we show that the \kl{string-to-string functions} computed by
\kl{protocols} share some  good properties of \kl{regular functions}, such as linear size
outputs and computability. Even computability is not a priori obvious, due to
the non-uniformity of the protocols. These results can be seen as evidence of
the open implication 
\begin{align*}
\text{protocol} \implies \text{regular}
\end{align*}
in the conjecture. To prove these results, we will leverage the results on
\kl{weighted automata} from Section~\ref{sec:field-domain}. The point of
departure is the following lemma, which connects \kl{weighted automata} and
\kl{string-to-string functions} that can be computed in our protocol.

    \begin{lemma}
        \label{lem:postcomposition-weighted-automaton}
        Let $\domain$ be a semiring, and consider two functions
        \[
        \begin{tikzcd}
        \Sigma^* 
        \ar[r,"f"]
        &
        \Gamma^*
        \ar[r,"g"]
        & 
        \domain,
        \end{tikzcd}
        \]
        such that $f$ is computed by a \kl{protocol} (with string outputs) and $g$ is computed by a \kl{weighted automaton}, then the composition  $f;g$ is computed by a protocol (with outputs in $\domain$).
    \end{lemma}
    \begin{proof}
      For each string in $\Gamma^*$, the \kl{weighted automaton} for $g$ has an
      associated matrix over the semiring $\domain$, as explained in the proof at the beginning of Section~\ref{sec:proof-of-thm-field-domain}.
      Similarly to that proof, we can modify the protocol for $f$ to a protocol for $f;g$ where the parties send matrices instead of strings. 
    \end{proof}

\begin{corollary}\label{cor:protocol-field-continuity}
    If the domain $\domain$ is a field, then the conclusion of \cref{lem:postcomposition-weighted-automaton} can be strengthened to say that $f;g$ is computed by a weighted automaton.
\end{corollary}
\begin{proof}
    For field outputs, protocols are equivalent to weighted automata, thanks to \cref{thm:field-domain}.
\end{proof}
The above corollary establishes a property of $f$, namely that \kl{weighted automata}
(over a field) are closed under precomposition with $f$. We think that this is
an important property, and therefore we give it a name.

\begin{definition}[Field continuity]
    \label{def:weighted-continuity}
    \AP
    A \kl{string-to-string function} $f : \Sigma^* \to \Gamma^*$ is called
    \intro{field continuous} if functions computed by weighted automata over a
    field are closed under precomposition with $f$.
\end{definition}

The name ``continuous'' is inspired by a
similar terminology that is used in automata theory for functions that preserve
regularity under inverse images, see~\cite[Theorem 4.1]{PinSilva05}
or~\cite[Footnote 2]{continuity20}. For the latter notion, we use the name
\emph{Boolean continuity}.

\begin{definition}[Boolean continuity]
  \AP
  A \kl{string-to-string function} $f : \Sigma^* \to \Gamma^*$ is called
  \intro{Boolean continuous} if preimages of \kl{regular languages} are \kl[regular language]{regular}.
\end{definition}

As we have shown in \cref{cor:protocol-field-continuity}, all
\kl{string-to-string functions} computed by \kl{protocols} are \kl{field continuous}. In
particular, since every \kl{regular string-to-string function} is computed by a
\kl{protocol}, it follows that every regular string-to-string function is field
continuous\footnote{To the best of our knowledge, this is a new result. It can
also be proved directly, without passing through protocols.}. We
conjecture that the converse is also true.

\begin{conjecture}\label{conj:regular-continuous}
  A \kl{string-to-string function} is \kl{field continuous} if and only if it is 
  \kl[regular function]{regular}.
\end{conjecture}

In Example~\ref{ex:quadratic-counterexample} later in this section,  we show that the conjecture becomes false if the
left side is relaxed from \kl{field continuous} to \kl{Boolean continuous}.

Conjecture~\ref{conj:regular-continuous} can be seen as a machine independent
characterisation of the \kl{regular string-to-string functions}. This would be
a very valuable contribution. Almost all known characterisations of the
\kl{regular string-to-string functions} have somewhat lengthy definitions,
based on specific computational models, and it is something of a miracle that
all of these models are equivalent. A possible exception is the
characterisation in~\cite{bojanczykTitoRegular23}, which does not use a machine
model; however that characterisation uses the abstract language of category
theory, and is less elementary than the one in
Conjecture~\ref{conj:regular-continuous}.

As in Conjecture~\ref{conj:protocol-regular-string-to-string}, the content of Conjecture~\ref{conj:regular-continuous} is the left-to-right implication
\begin{align*}
\text{field continuous} \implies \text{regular}.
\end{align*}
Conjecture~\ref{conj:regular-continuous} is stronger than Conjecture~\ref{conj:protocol-regular-string-to-string}, as explained in the following diagram, which shows the known relations between three kinds of string-to-string functions:
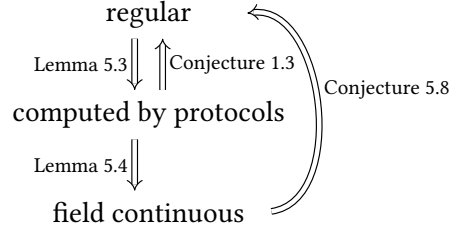
\[
\begin{tikzcd}
\text{~~~~regular~~~~}
\ar[d,Rightarrow,shift right=2, "\text{\cref{lem:from-regular-to-protocol}}"']
\\
\text{computed by protocols}
\ar[d,Rightarrow, shift right=2, "\text{\cref{lem:postcomposition-weighted-automaton}}"']
\ar[u,Rightarrow, shift right=2,"\text{Conjecture~\ref{conj:protocol-regular-string-to-string}}"']
\\ 
\text{field continuous} 
\ar[uu,bend right=89, Rightarrow, shift right=2,"\text{Conjecture~\ref{conj:regular-continuous}}"']
\end{tikzcd}
\]

The following theorem gives some evidence for the stronger conjecture,  and
therefore also the weaker one, by showing that the field continuous functions
share some well-known properties of the regular string-to-string functions. 

\evidencefortheconjecture*


\begin{proof}
    For properties~\ref{it:linear-size-outputs}
    and~\ref{it:linear-time-computable}, we embed strings into numbers. An
    output string over alphabet $\Gamma$ can be seen as a number in base
    $|\Gamma|$. To avoid the ambiguity that could result from leading zeros, we
    first prepend the string with the digit 1. Let 
    \begin{align*}
    g : \Gamma^* \to \Nat \subseteq \Rat
    \end{align*} 
    be the corresponding encoding. This encoding can be computed by a \kl{weighted
    automaton} over the field $\Rat$, see~\cite[Lemma 8.10]{bojanczyk_automata_2025}.
    By the assumption on \kl{field continuity}, the
    composition $f;g$ can be computed by a \kl{weighted automaton}. This is a
    \kl{weighted automaton} that works in the field of rationals $\Rat$, but  only
    produces natural numbers on its output. By~\cite[p.
    110]{BerstelReutenauer08} the automaton can be chosen so that it only uses
    integers $\Int$, possibly including negative integers. Summing up, we have
    a \kl{weighted automaton} over $\Int$ that outputs the representation, in base
    $|\Gamma|$, of the output string produced by $g$. We claim that for such an
    automaton, the output number
    \begin{enumerate}
        \item has a linear number of digits;
        \item can be computed in linear time.
    \end{enumerate}
    These two claims yield the corresponding items in the statement of the
    theorem. The first item, about a linear number of digits, is true because
    it is true for every \kl{weighted automaton} over $\Int$. This is because
    applying a fixed linear map can only add a constant number of digits, and therefore each input letter can increase the output by a constant number of digits. For the second item, use the fact that a weighted automaton can be evaluated in linear time. We did not find this result in the literature, so we give a proof sketch below.

    \begin{claim}
        Fix a function $f : \Sigma^* \to \Int$ computed by a weighted automaton over $\Int$. Then $f$ can be computed in linear time (ignoring logarithmic factors).
    \end{claim}
    \begin{proof}
        We use the \textsc{ram} model of computation.  In this model, numbers with $n$ digits can be multiplied in time linear in $n$ (ignoring logarithmic factors), using Fast Fourier Transform~\cite[p.291]{schonhage1971schnelle}. If we would try to evaluate the weighted automaton naively, from left to right, we would be doing a linear number of multiplications, involving numbers with a linear number of bits. This would give a quadratic time algorithm. To get linear time, we use  divide and conquer approach which was suggested to us by Marek Soko{\l}owski, and which is inspired by similar algorithms for multipoint evaluation of polynomials~\cite[Section 10.1]{von2003modern}.

        For each input letter, the weighted automaton has an associated matrix. Therefore, the essence of the  problem is to compute the multiplication of $n$ matrices $
        A_1 \cdots A_n$,
        which are taken from some fixed finite set of square matrices of same dimension. As we have already remarked before, if we multiply $n$ matrices, then each entry in the resulting matrix has a $\Oo(n)$ bits. Instead of multiplying the matrices from left to right, we organise the multiplication in a balanced binary tree, and evaluate this tree bottom-up. At the bottom of this tree, we have $n/2$ multiplications of matrices from the finite set. More generally, in each subtree of height $h$, the resulting matrix has entries with $\Oo(2^h)$ bits, and there are $n / 2^h$ such subtrees. Therefore, assuming that we have evaluted the value of each subtree of height $h-1$, then we can evaluate each subtree of height $h$, by using $\Oo(n/2^h)$ instances of matrix multiplication, with the ocrresponding matrices having entries that have  $\Oo(2^h)$ bits. By using algorithms for addition and multiplication that are linear (ignoring logarithmic factors), the passage from height $h-1$ to height $h$ can be done in linear time (ignoring logarithmic factors). Since the number of heights is  logarithmic, we get the desired result.
    \end{proof}


    We are left with property~\ref{it:regular-preimages}, about \kl{Boolean
    continuity}. This will follow from the special case of \kl{field continuity},
    where the field is the two-element field.  This is because of  the
    following  folklore correspondence between regular languages and weighted
    automata over the two-element field. 
        
        \begin{claim}\label{claim:regular-weighted-automata}
            A language $L \subseteq \Gamma^*$ is regular iff its characteristic
            function $\Gamma^* \to \set{0,1}$ is computed by a \kl{weighted
            automaton} over the two-element field 
        \end{claim}
        \begin{proof}
            For the left-to-right implication, we observe that a weighted
            automaton over a finite field can be simulated by a deterministic
            finite automaton. For the other direction, we observe that a
            weighted automaton can count the parity of  the number of runs in a
            finite automaton, and if the automaton is deterministic then the
            number of runs is either zero or one, and thus the parity gives the
            right answer.
        \end{proof}

        In terms of the correspondence from the above claim, preimages of
        regular languages become precompositions of \kl{weighted automata} over the
        two-element field. In particular, regularity is preserved. 
\end{proof}

One could think that already the three properties in the above theorem are not
only necessary for regularity, but also sufficient. This is not the case, as
shown by the following example.

\begin{myexample}[Factorials]
    \label{ex:not-regular-but-continuous-over-finite-fields}
    Consider a \kl{string-to-string function}
    \begin{align*}
    g : \Sigma^* \to \set{a}^*
    \end{align*}
    where both the input and output alphabets are unary. A sufficient condition 
    for \kl{Boolean continuity} of such functions is given in~\cite[Example 2.12]{bojanczykTitoRegular23},
    using \emph{factorials}, i.e.~numbers in the set $\setbuild{n!}{$n \in \Nat$}$. 
    This sufficient condition is that: 
    (a) every output string arises from finitely many inputs; and
    (b) every output string has length that is a factorial. 

    It is not hard to come up with a non-regular function that has properties
    (a) and (b), thus ensuring \kl{Boolean continuity}, and which has furthermore
    linear size outputs and is computable in linear time. For example, the
    function could map an input string $w$ to the longest string of factorial
    length that is shorter than $w$. 
\end{myexample}

\subsection{Unary output alphabet}
\label{sec:unary-output-alphabet}
\AP
In this section, we provide further evidence for
Conjectures~\ref{conj:protocol-regular-string-to-string}
and~\ref{conj:regular-continuous},
 by showing that they are true for output
alphabets with only one letter. 
To prove the conjectures, we will pass through rational functions, which are a fragment of the regular functions. When the output alphabet has only one letter, rational functions and regular functions coincide, which means that the results on rational functions are relevant in this case.

\subsubsection{Rational functions and monotone protocols}
In this section, we use a variant of the  protocols  to characterise exactly the rational string-to-string functions.  This will be done by applying a characterisation of the rational functions due to Sch\"utzenberger and Reutenauer~\cite{reutenauerSchutzenberger1991}, which happens to be ideally matched to our protocol model, subject to a monotonicity restriction. Let us begin by defining the relevant notions.

\paragraph*{Rational functions.} There are several ways of defining the rational string-to-string functions. For the results in this section, the particular choice of definition will not be important, since we will refer to an external result about rational functions. Nevertheless, for the reader's convenience we include a definition, which is based on weighted automata over the semiring of finite sets of words. 

 Consider a nondeterministic automaton with an input alphabet $\Sigma$, together with an output map 
\begin{align*}
\text{output} : 
\myunderbrace{(I + F + \Delta)}{disjoint union of initial states, final states, and transitions}
\to \Gamma^*.
\end{align*}
Using the output map,  we can associate to each accepting  run of the automaton an output string, which is obtained by concatenating the output strings associated with the initial state,  the transitions used in the run, and the  final state. This way, the automaton defines a relation between input and output strings, which maps an input string to the outputs of all possible accepting runs. Such a relation is called a  \emph{rational string-to-string relation}. A \emph{rational string-to-string function} is the special case where each input string is mapped to exactly one output string. 

\begin{myexample}
  The identity function is rational, and so is the function which deletes all $a$'s from the input string. Another example is the function which swaps the first and last letters, and leaves the rest of the input string unchanged. String reversal is not rational; also duplication is not rational.
\end{myexample}

\paragraph*{Monotone protocols.}
A protocol with string-to-string outputs is called \emph{monotone} if the output string is produced by concatenating two strings, the left one produced by Alice, and the right one produced by Bob. One can consider monotone protocols with one or more rounds. The reduction to one-round protocols from Lemma~\ref{lemma:one-round-reduction-general} also works in the presence of the monotone restriction, and therefore we will only talk about one-round monotone protocols. Such a protocol can be described more elementarily as follows. The two players have strategies 
\begin{align*}
    \sigma_A & : \Sigma^* \to Q_A \times (\Gamma^*)^{d},\\
        \sigma_B & : \Sigma^* \to Q_B \times (\Gamma^*)^{d},
\end{align*}
in which they send a signal as well as a tuple of $d$ strings over the output alphabet. Based on the signals received from both players, the output is produced by concatenating two of the output strings, one from Alice and one from Bob. This is formalised by  an output function of type
\begin{align*}
  \text{output} : Q_A \times Q_B \to \set{1,\ldots,d}.
\end{align*}
The output of the protocol is then defined as follows: we apply the output function to the signals produced by both players, which gives us an index $i \in \set{1,\ldots,d}$, and then the output string is obtained by concatenating the $i$-th string from Alice with the $i$-th string from Bob (thus ensuring that Alice's part of the output always comes before Bob's).

\paragraph*{Equivalence of the models.} We now show that the monotone protcols compute exactly the rational string-to-string functions. As it turns out, this equivalence was essentially already proved by Sch\"utzenberger and Reutenauer in~\cite{reutenauerSchutzenberger1991}, albeit in a different language. The corresponding characterisation is stated below.

\begin{theorem}
  \label{thm:monotone-protocols-rational}
  For a string-to-string function $f : \Sigma^* \to \Gamma^*$, the following are equivalent:
  \begin{enumerate}
    \item \label{item:monotone-protocol} $f$ is computed by a monotone protocol;
    \item \label{item:partial-functions} there is a finite family of partial string-to-string functions, indexed by a finite set $I$
\begin{align*}
\set{\alpha_i,\beta_i  : \Sigma^* \to \Gamma^*}_{i \in I}
\end{align*}
such that the (total) function 
\begin{align*}
  (w_1,w_2) \quad \mapsto \quad f(w_1 w_2)
\end{align*}
is equal to  the union of partial functions
\begin{align*}
  \bigcup_{i \in I} \Big( (w_1,w_2) \mapsto \myunderbrace{\alpha_i(w_1) \cdot \beta_i(w_2)}{the concatenation is defined \\  only if both functions are defined} \Big).
\end{align*}
    \item \label{item:rational-function} $f$ is a rational string-to-string function.
  \end{enumerate}
\end{theorem}
\begin{proof}
The equivalence of items \ref{item:partial-functions} and \ref{item:rational-function}, which is the essence of the theorem, was proved  by Sch\"utz-enberger and Reutenauer in~\cite[p. 674] {reutenauerSchutzenberger1991}. It remains to show the equivalence of items \ref{item:monotone-protocol} and \ref{item:partial-functions}, which is a simple unfolding of the definitions.

\begin{itemize}
  \item 
    \ref{item:partial-functions} $\implies$ \ref{item:monotone-protocol}. Suppose that $f$ satisfies item~\ref{item:partial-functions}. 
    As her signal, Alice says which of the functions $\alpha_i$ are defined on her part of the input, and sends all the corresponding outputs. Bob does the same on his side. Based on the signals, they can determine an index $i$ such that both $\alpha_i$ and $\beta_i$ are defined, and then they can concatenate the corresponding outputs to produce the final output. This shows that item~\ref{item:monotone-protocol} holds.
  \item
    \ref{item:monotone-protocol} $\implies$ \ref{item:partial-functions}. Suppose  that $f$ satisfies item~\ref{item:monotone-protocol}.
    Based on the monotone protocol from the assumption, we define the indexing set  $I$ to be  $Q_A \times Q_B$, i.e.~all possible combinations of signals of Alice and Bob. For each such pair of signals $(q_A,q_B) \in Q_A \times Q_B$, we define two partial functions. The first one  is defined on those strings $w_1$ for which Alice's signal is $q_A$, and in this case it produces the $i$-th string from Alice, where $i$ is obtained by applying the output function to $(q_A,q_B)$. The second one is defined similarly on Bob's side. It is easy to see that the conditions from item~\ref{item:partial-functions} are now satisfied.
\end{itemize}
\end{proof}

Before continuing, let us remark that a similar approach could be taken to characterise the subsequential functions, which are the subclass of rational functions that corresponds to deterministic automata. In this case, the corresponding restriction on the protocols is  monotonicity, plus the extra requirement  is that the output function does not depend on the signal sent by Bob, i.e.~the only signal information travels from left to right. Using a characterisation of subsequential functions from~\cite{choffrut1977}, one can show that such protocols compute exactly the subsequential functions. The details are left to the reader.

\subsubsection{Proofs of Conjectures~\ref{conj:protocol-regular-string-to-string}
and~\ref{conj:regular-continuous} for a unary output alphabet}
Using the above characterisation of rational functions, we can now prove Conjectures~\ref{conj:protocol-regular-string-to-string} and~\ref{conj:regular-continuous} for unary output alphabets. The first conjecture is an immediate corollary. 
\begin{theorem}\label{thm:unary-string-to-string}
  Conjecture~\ref{conj:protocol-regular-string-to-string} holds for string-to-string functions where the output alphabet has only one letter. In other words, for an output alphabet $\Gamma$ with one letter, protocols compute exactly the regular string-to-string functions. 
\end{theorem}
\begin{proof}
  When the output alphabet has only one letter, then: (a) rational string-to-string functions coincide with the regular string-to-string functions, which is a folklore result; and (b) monotone protocols coincide with general protocols, since there is no difference between concatenating two strings in different orders. Therefore, the equivalence from \cref{thm:monotone-protocols-rational} implies the theorem.
\end{proof}

In the rest of this section, we use the above theorem to derive some further observations, and in particular to prove Conjecture~\ref{conj:regular-continuous} for unary output alphabets. The first observation is about protocols that use integers with addition and have non-negative outputs.


\begin{theorem}
  \label{thm:string-to-number-protocols-nat}
  Consider a protocol with output domain $(\Int,+)$ where all outputs are non-negative. Then the same function can be computed by a protocol with output domain $(\Nat,+)$.
\end{theorem}
\begin{proof}
  Let us assume without loss of generality that $f$ is computed by a one-round
  protocol.  We use the name \emph{configuration} for the message sent by alice, similarly to the proof of \cref{thm:unary-string-to-string}. The configuration consists of a signal, and several integers. We can improve the protocol so that:  (a) for each input with a split $w = w_1 w_2$ exactly one of the numbers in Alice's configuration is part of the  final output; and (b) if Alice's configuration for a string $w_1$ uses some integer, then  there is some $w_2$ where this integer contributes to the final output.  These assumptions can be ensured by a powerset
  construction on Alice's side.

  Let us prove the following intermediate claim: the integers in  Alice's configuration are
  uniformly bounded below. This is a simple corollary of the properties (a) and (b): if the integers would be arbitrarily large negative numbers, then the contribution of Bob for some fixed string $w_2$ would not be enough to offset them.

  Thanks to the above claim, we can further improve the protocol as follows. Alice sends her original configuration, with the following change: if some numbers were negative, then they are truncated to zero, and the signal is adjusted so that it stores the negative numbers. This can be done in a finite signal space, due to the claim on lower bounds. Once Bob receives the message, he can the produce the output of the original protocol, by possibly subtracting Alice's negative part encoded in the signal from his part of the output.
\end{proof}

This theorem allows us to derive a characterisation of functions 
computed by weighted automata over $\Nat$ that have linear growth.

\begin{corollary}
  \label{cor:weighted-automata-nat-regular}
  Let $f : \Sigma^* \to \Rat$ be a function computed by a weighted automaton
  over  $\Rat$, which has linear growth (i.e.~the output number is at most linear in the input length), and such that all outputs are natural numbers. Then $f$ is a regular function, when viewed as a string-to-string function with a unary output alphabet.
\end{corollary}
\begin{proof}
  It is well-known that weighted automata over $\Rat$ computing integer values
  can be simulated by weighted automata over the semiring $\Int$: this is
  sometimes referred to as $\Rat$ being a Fatou extension of $\Int$
  \cite[p. 110]{BerstelReutenauer08}.
  Therefore, we can assume without loss of generality that $f$ is computed by
  a weighted automaton over $\Int$. Then, one can leverage a 
  result from \cite{Zpolyreg23} stating that weighted automata over $\Int$
  with linear growth can be computed as follows: first, a regular function $g : \Sigma^* \to
  \{a,b\}^*$ is computed, then each letter $a$ is mapped to $+1$ and each 
  letter $b$ is mapped to $-1$, and finally all values are summed up. 
  For a precise statement, see \cite[Proposition II.13, Theorem III.3]{Zpolyreg23}.
  Now, it is clear from the above description that $f$ can be computed by a protocol
  manipulating integer values, and combining them only using addition. By
  \cref{thm:string-to-number-protocols-nat}, $f$ can also be computed by a
  protocol manipulating natural numbers. Finally, by \cref{thm:unary-string-to-string},
  $f$ is a regular function.
\end{proof}

From the above, we conclude that Conjecture~\ref{conj:regular-continuous} holds for unary output alphabets.
\begin{corollary}
  \label{cor:regular-continuous-unary}
  Conjecture~\ref{conj:regular-continuous} holds for string-to-string functions where the output alphabet has only one letter.
  In other words, a string-to-string function with unary output alphabet is field continuous if and only if it is regular.
\end{corollary}
\begin{proof}
  By \cref{cor:protocol-field-continuity}, it suffices to show the only-if direction. 
  Let $f : \Sigma^* \to \Gamma^*$ be a function with a unary output alphabet which is field continuous. Compose this function with the weighted automaton $\Aa : \Gamma^* \to \Rat$ which computes the length of the output string. The result is a function that is subject to the assumptions of \cref{cor:weighted-automata-nat-regular}, and therefore it is a regular function. 
\end{proof}
Let us conclude this section with an example showing that
\cref{cor:weighted-automata-nat-regular} cannot be extended to weighted
automata over $\Int$ with arbitrary growth.

\begin{myexample}[Quadratic counterexample]\label{ex:quadratic-counterexample}
     We show  a function which: (a) is a linear combination of \mso  counting functions of arity two, with negative coefficients; (b) has only non-negative outputs; and (c) cannot be presented as linear combination with positive coefficients. The idea, which is based on~\cite[Example 2.1]{BerstelReutenauer08}, is to trivially ensure non-negativity by squaring. Take the function
\begin{align*}
w \in \set{a,b}^* 
\quad \mapsto \quad 
\left(\text{(number of $a$'s in $w$)} - \text{(number of $b$'s in $w$)}\right)^2.
\end{align*}
This function clearly satisfies (a) and (b). As explained in \cite[p.3]{Zpolyreg23}, it also satisfies (c), since the inverse image of $0$ is not a regular language, as would be the case if only positive coefficients were used.
\end{myexample}

\section{Infinite alphabets}
\label{sec:infinite-alphabets}
\AP
In this section, we present a variant of our model which deals with an input
alphabet. This direction is rooted in the tradition of language theory for
infinite alphabets, which dates back to the work of Kaminski and
Francez~\cite{kaminskiFiniteMemoryAutomata1994}, and has been developed in many
subsequent papers, see e.g.~the
survey~\cite{bojanczykOrbitFiniteSetsTheir2017}. The general idea is that we
have an infinite alphabet $\atoms$ (whose elements are called \intro{atoms}), and the
languages  refer only to equality between letters, as in the following examples
\begin{align}
\setbuild{ w \in \atoms^*}{the first letter is equal to the last letter}
\label{eq:first-last}
\\
\setbuild{ w \in \atoms^*}{some letter appears at least twice}
\label{eq:some-twice}
\end{align}

There are numerous models of automata for such languages, which typically
involve some kind of finite memory, as well as registers that store letters
from $\atoms$. For example, the language~\eqref{eq:first-last} is recognised by
an automaton which loads the first letter into a register, and then toggles
acceptance depending on comparison of the register with the current input
letter. The language~\eqref{eq:some-twice} is recognised by an automaton which
nondeterministically guesses a position, puts its letter into a register, and
then waits for this letter to appear again. 

Numerous models for infinite alphabets have been proposed in the literature,
see Figure~\ref{fig:automata-infinite-alphabets} which contains a sample of
seventeen models. Interestingly, all of those models are pairwise
non-equivalent. This sharply contrasts with the finite-alphabet case, where
virtually all models coincide and capture the regular languages. 

In this section, we describe an infinite-alphabet version of our two-party
protocols. The motivation for this study is twofold: (a) a search for a
canonical model of regular languages for infinite alphabets; and (b)
mathematical interest. Regarding the point (a), we hope that the adaptability
of two-party protocols to various settings will help us  find a canonical model
of regular languages for infinite alphabets. This seems to be at least
partially successful, since there is evidence --- which we present in this
section --- that the protocols are equivalent to one of the existing automaton
models\footnote{
  Thus avoiding proliferation of standards, humanistically depicted in a famous XKCD comic.
}, namely unambiguous register automata, see item~\ref{it:orbit-finite-unamb} in
Figure~\ref{fig:automata-infinite-alphabets}. If true, this
equivalence would be unexpected, since there does not seem to be any syntactic
connection between unambiguous register automata and protocols. Regarding point
(b), one of the exciting features of our protocol model for infinite alphabets
is that the interaction between the two parties becomes essential, and the
protocol cannot be reduced to the one-round case as in
\cref{lemma:one-round-reduction-general}.

\begin{figure}
    \begin{enumerate}
    \item deterministic register automata~\cite[Definition 3]{kaminskiFiniteMemoryAutomata1994}
    \item nondeterministic register automata~\cite[Definition 1]{kaminskiFiniteMemoryAutomata1994}
    \item nondeterministic register automata with guessing~\cite[Definition 2.7]{bojanczyk_slightly}
    \item weighted register automata over the two-element field~\cite[Definition 3.1]{orbitFiniteVectorTheoretics}
    \item two-way deterministic register automata~\cite[Definition 5]{kaminskiFiniteMemoryAutomata1994}
    \item two-way nondeterministic register automata~\cite[Definition 2.1]{nevenFiniteStateMachines2004}
    \item alternating register automata~\cite[p.~16:8]{lazicDemri09}
    \item alternating register automata with one register~\cite[p.~16:19]{lazicDemri09}
    \item \label{it:orbit-finite-unamb} unambiguous register automata~\cite[Section 5]{colcombet2015unambiguity}
    \item register automata with pebbles~\cite[Section 2.2]{nevenFiniteStateMachines2004}
    \item \label{it:single-use} single-use register automata~\cite[Definition 2]{bojanczykstefanski2020}
    \item data automata~\cite[Section 4.2]{bojanczykTwovariableLogicData2011}
    \item class automata~\cite[Section III]{bojanczykExtensionDataAutomata2010} 
    \item regular expressions~\cite[Definition 2]{regexpKaminskiTan2004}
    \item three other kinds of regular expressions~\cite[Sections 4, 5, 6]{regexpLibkin2015}
    \item yet another kind of regular expressions~\cite[Section 5]{KleeneNominal2019}
    \item monadic second-order logic with equality~\cite[Section 2.4]{nevenFiniteStateMachines2004}
\end{enumerate}
    \caption{A non-exhaustive list of models of automata for infinite alphabets. All models in the list are pairwise non-equivalent. In contrast, for finite alphabets, all models in this list are equivalent, and define exactly the regular languages. }
    \label{fig:automata-infinite-alphabets}
\end{figure}


\subsection{Protocols for an infinite alphabet}
\label{sec:protocols-infinite-alphabet}
\AP
We now give a more detailed description of our model. 
As explained in the introduction, we only care about languages that are closed 
under \intro{permutations of the alphabet}, according to the following definition. 

\begin{definition}[Equivariant language] \label{def:equivariant-language}
  \AP
  A language $L \subseteq \atoms^*$ is called \intro[equivariant language]{equivariant} if 
  \begin{align*}
  w \in L \quad \iff \quad \pi(w) \in L
  \end{align*}
  holds for every permutation $\pi$ of the alphabet $\atoms$.
\end{definition}

Examples of \kl{equivariant languages} include the languages
in~\eqref{eq:first-last} and~\eqref{eq:some-twice}. On the other hand, the
language ``the first letter is a vowel'' or ``the letters are strictly
increasing'' are not \kl{equivariant}, since there is no such thing as a vowel,
or an ordering of the letters. The principle of \kl[equivariant
language]{equivariance} will also be applied to protocols, as described below.

\AP To define \intro{(simple) equivariant protocols} for infinite alphabets, we use the
same kind of \kl{protocols} as in Definition \ref{def:two-party-protocol-boolean},
except that apart from bits, the parties can also send letters from the
alphabet $\atoms$. It follows that the allowed set of messages is
$\set{\text{true, false}} + \atoms$, i.e.~the disjoint union of the Booleans
and the input alphabet. Similarly, to
Definition~\ref{def:two-party-protocol-boolean}, the number of rounds is a
fixed number $k$ and in the $i$-th round each party chooses a new message
according a strategy which is a function of the following type: 

\begin{align*}
\myunderbrace{\atoms^*}{local \\ string} \times \myunderbrace{(\set{\text{true, false}} + \atoms)^{i-1}}{messages received \\ in previous rounds} 
\to
\myunderbrace{\set{\text{true, false}} + \atoms}{message sent \\ in this round}
\end{align*}
In the last round, Bob must send a bit, and this bit is the output of the \kl[simple equivariant protocol]{protocol}.
An important restriction in the protocol is that the strategies of both parties must be \emph{equivariant} in the following sense: 
a strategy $\sigma$ is equivariant if for every permutation $\pi$ of $\atoms$ and for every $i$, it satisfies the following condition:
\begin{align*}
\sigma(w, m_1, \ldots, m_{i-1}) = m_i 
\quad \Rightarrow \quad
\sigma(\pi(w), \pi(m_1), \ldots, \pi(m_{i-1})) = \pi(m_i).
\end{align*}
In the above, $\pi$ is applied to messages in the natural way -- 
it modifies the \kl{atoms} and leaves the Booleans unchanged.

\begin{myexample}[Repetitions cannot be detected]
\label{ex:protocol-not-repetitions}
    Having formally defined the protocols for infinite alphabets,
    we can now revisit Example~\ref{ex:reg-ndet-too-strong} from the introduction
    and prove that the language ``some letter appears at least twice'' cannot be computed by a
    \kl[simple equivariant protocol]{protocol}. 
    Suppose, towards a contradiction, that there is a protocol with
    $k$ rounds that computes this language. Consider an input string with
    $2k+2$ pairwise different letters, split so that Alice and Bob get $k+1$
    letters each. In the execution of the protocol there are at most $k$ \kl{atoms}
    which are sent as messages. In particular, there must be some atom $a$ that
    appears in Alice's part of the input string, but is not sent as a message,
    and similarly there must be some atom $b$ that appears in Bob's part of the
    input string, but is not sent as a message. Consider an \kl{atom permutation}
    $\pi$ which swaps $a$ with $b$. If we apply this atom permutation to
    Alice's part of the string (but not Bob's), then the communication history
    will remain unchanged. In particular, the output of the protocol will be
    the same on both inputs. However, after applying this permutation, the
    input string has a repetition, unlike the original one. 
\end{myexample}

\subsubsection{Orbit-finite sets}
\label{sec:orbit-finite-sets}
\AP
In this section, we do a more systematic analysis of automata models for
infinite alphabets, and their relationship to our protocols. As an organising
principle, we use the approach of orbit-finite sets,
see~\cite{bojanczykOrbitFiniteSetsTheir2017} for a longer survey, which is a
generalisation of finite sets suitable for infinite alphabets sets. Using this
notion, we can lift any model of computation that uses finite sets, to a model
that uses orbit-finite sets, which allows us for a clean comparison of the two
setting. For the purposes of this paper, we use a simplified definition of
orbit-finite sets, which is sufficient for our purposes\footnote{
Definition~\ref{def:orbit-finite-sets} is weaker than the usual notion of orbit-finite sets~\cite[Section 5]{bojanczyk_slightly}; in fact it is the special case of the usual notion that is called \emph{polynomial orbit-finite sets} in~\cite[Section 1]{bojanczyk_slightly}.
The stronger  notion that is usually used  allows for two extra features: (a) restricting to equivariant subsets (e.g.~one could limit $\atoms^2$ to pairs which are non-repeating); and (b)  symmetries (e.g.~one could identify pairs in $\atoms^2$ if they agree up to swapping of coordinates, thus yielding unordered pairs). In some cases, the extra features are desirable, in particular they establish a connection with set theory~\cite{blassDedekind2016} and  nominal sets~\cite[Section 5]{PittsAM:nomsns}. However, those features do not play any role in the analysis of protocols and automata and, to avoid technicalities, we use the simpler polynomial version from Definition~\ref{def:orbit-finite-sets}. This simplification is purely technical --- all results continue to hold for the usual notion of orbit-finite sets.
}.
\begin{definition}[Orbit-finite sets] \label{def:orbit-finite-sets}
  \AP
  An \intro{orbit-finite set} is any set of the form 
    \begin{align*}
    \atoms^{d_1} + \cdots + \atoms^{d_n},
    \end{align*}
    for some natural numbers $d_1,\ldots,d_n \in \set{0,1,\ldots}$. 
\end{definition}

\AP
When a summand $\atoms^{d_i}$ uses $d_i =0$, then it  describes a set with
exactly one element, namely the empty tuple. Therefore, \kl{orbit-finite sets}
generalise finite sets, since a  finite set with $n$ elements can be seen as
the orbit-finite set which has $n$ disjoint copies of $\atoms^0$. We will only
be interested in subsets of orbit-finite sets and functions between them that
are \intro{equivariant}, i.e.~invariant under permutations of the \kl{atoms}, in
the following sense:
\begin{align*}
\myunderbrace{x \in X \iff \pi(x) \in X}{equivariant subset $X$ \\ of an orbit-finite set}
\qquad 
\myunderbrace{f(x) = y \iff f(\pi(x)) = \pi(y)}{equivariant function $f$ \\ between two orbit-finite sets}
\end{align*}
We can now discuss various orbit-finite models of computation, by generalising
finite sets to orbit-finite ones, and requiring all subsets and relations to be
\kl{equivariant}. As a first example of this approach, we can revisit the definition
of \kl{simple equivariant protocols} from Section~\ref{sec:protocols-infinite-alphabet}, and
define it in terms of \kl{orbit-finite sets}:

\begin{definition}[Orbit-finite protocol]
    \label{def:orbit-finite-protocol}
    An \intro{orbit-finite Boolean two-party protocol}
    is defined in the same way as in Definition \ref{def:two-party-protocol-boolean}, except that:
  \begin{enumerate}
    \item the input alphabet $\Sigma$, and the message spaces $Q_A$ and $Q_B$ are \kl{orbit-finite}; and 
    \item the \kl{strategies} of both players and the output function are \kl{equivariant}.
  \end{enumerate}
\end{definition}

Indeed, the \kl{simple equivariant protocols} from
Section~\ref{sec:protocols-infinite-alphabet} are the special case of the above
definition where the input alphabet is $\atoms$, and the message spaces are
both equal to 
\begin{align*}
\myunderbrace{\atoms}{letter} + \myunderbrace{\atoms^0}{bit 0} + \myunderbrace{\atoms^0}{bit 1}.
\end{align*}
On the other hand, the \kl[simple equivariant protocol]{special case} is also
equivalent to the \kl[orbit-finite protocols]{general case}, since an element of a general \kl{orbit-finite set}
can be transmitted using a constant number of bits and \kl{atoms} (by first sending
the index of the summand, and then sending the tuple of atoms). It follows that
the protocols from Definition~\ref{def:orbit-finite-protocol} and those from
Section~\ref{sec:protocols-infinite-alphabet} have the same expressive power.
From now on we will use the formalisation from
Definition~\ref{def:orbit-finite-protocol}.

Orbit-finiteness can also be used to define automata. The following definition
has the same expressive power as the standard (nondeterministic and
deterministic) register automata for infinite alphabets
from~\cite{kaminskiFiniteMemoryAutomata1994}; this equivalence was shown
in~\cite[Lemma 6.3]{bojanczykAutomataTheoryNominal2014} and is one of the
original motivations for studying orbit-finiteness.

\begin{definition}
    [Orbit-finite automata]
    \label{def:orbit-finite-automata}
    \AP
    A \intro{nondeterministic orbit-finite automaton} is defined in the same way as a 
    nondeterministic finite automaton, except that all sets are orbit-finite, 
    and all subsets and functions are \kl{equivariant}:
\begin{align*}
    \myoverbrace{
        \myunderbrace{Q}{states} \quad 
        \myunderbrace{\Sigma}{input \\ alphabet}
    }
    {orbit-finite}
    \qquad
    \myoverbrace{
        \myunderbrace{I \subseteq Q}{initial \\ states} \quad 
        \myunderbrace{F \subseteq Q}{final \\ states} \quad 
        \myunderbrace{\Delta \subseteq Q \times \Sigma \times Q}{transitions}
    }{equivariant}.
\end{align*}
A \intro{deterministic orbit-finite automaton}
is the special case which has exactly one initial state, and where the transition relation is a function.
\end{definition}

As stated in Figure~\ref{fig:automata-infinite-alphabets}, deterministic and
nondeterministic orbit-finite automata have different expressive power.
Moreover, as stated in Examples~\ref{ex:reg-det-too-weak}~and~\ref{ex:reg-ndet-too-strong}
none of these models is equivalent to \kl{orbit-finite protocols}:
\kl[of dfa]{deterministic automata} are too weak, and 
\kl[of nfa]{nondeterministic automata} are too
strong. 

In Definition~\ref{def:orbit-finite-automata}, we have defined \kl{one-way
orbit-finite automata}, which read the input string from left to right. A
natural extension are the two-way automata, which can move their reading head
in both directions according to their transition function. This extension is
particularly natural in the context of protocols, with their two-way
interaction between the communicating parties. However, in the orbit-finite
setting, two-way automata are very strong:

\begin{myexample}[Two-way too strong]\label{ex:protocol-not-2dofa}
    The language ``some letter appears at least twice'' can also be recognised
    by a deterministic two-way orbit-finite automaton~\cite[Example
    18]{bojanczyk_slightly}. Therefore, this automaton model cannot be
    simulated by protocols.

    The reason why two-way orbit-finite automata are so strong is that they can
    make an unbounded number of visits to any given position -- for example the
    automaton for the language ``some letter appears at least twice'' will
    visit the last position a linear number of times (where the length of its run is
    quadratic). One idea to tame this power is to consider the bounded-crossing
    variant of two-way automaton, which has a fixed bound $k$ on the number of
    times that the automaton can visit a position \cite[p.~92]{neven2003power}.
    We believe that this model can actually be equivalent to the protocols.
    However, in our conjecture, we will focus on the better studied model
    presented in the following section.
\end{myexample}

\subsection{Unambiguous orbit-finite automata}
\label{sec:unambiguous-orbit-finite-automata}
\AP
As we have shown in Examples~\ref{ex:reg-ndet-too-strong}, \ref{ex:reg-det-too-weak}
and \ref{ex:protocol-not-2dofa},
\kl[of protocol]{protocols} are not equivalent to one-way deterministic or
nondeterministic orbit-finite automata, or their two-way variants automata. So
what is the right automaton model?  We conjecture that the answer is
\intro{unambiguous orbit-finite automata}, i.e.~the special case of
\kl{nondeterministic orbit-finite automata} that have zero or one accepting runs for
every input string.
\begin{conjecture}
    \label{conj:protocols-unambiguous}
    A language over an \kl{orbit-finite} alphabet is computed by 
    an \kl{orbit-finite protocol} if and only if it is recognised by an \kl{unambiguous orbit-finite automaton}.
\end{conjecture}

One corollary of this conjecture would be that \kl{unambiguous orbit-finite
automata} are closed under complement, since \kl[of protocol]{protocols} can
be complemented by flipping the output bit. This corollary has been conjectured
in~\cite[p.9]{colcombet2012forms}, and to the best of our knowledge remains
open, despite apparent claims to the contrary in~\cite[Footnote
5]{colcombet2015unambiguity}.

In this section, we prove the $\impliedby$ implication of the conjecture, i.e.~we
show that \kl{orbit-finite protocols} can simulate \kl{unambiguous orbit-finite
automata}. Unlike similar results earlier in this paper, this simulation is
non-trivial. Also, despite the one-way nature of the automata, the simulation
crucially depends on the interactive nature of protocols, i.e.~it requires more
than one round of communication. In particular multi-round protocols cannot be
reduced to one round, as was the case for finite alphabets
(i.e.~\cref{lemma:one-round-reduction-general} is no longer true for the
orbit-finite case).

\begin{theorem}
    \label{thm:unambiguous-to-protocol}
    If a language $L$ over an \kl{orbit-finite} alphabet is recognised by 
    \kl{an unambiguous orbit-finite automaton}, then 
    it is also computed by an \kl{orbit-finite protocol}.
\end{theorem}
\begin{proof}
  For the rest of this proof fix an \kl{unambiguous orbit-finite automaton}, whose state space is the \kl{orbit-finite} set $Q$.
Suppose that the input string is factorized as $w = w_1 w_2$. The general idea of the protocol is that Alice and Bob will jointly
compute the intermediate state (if it exists), i.e.~the state $q$ which satisfies:
\begin{align*}
\myunderbrace{I \xrightarrow{w_1}q}{there is a run over $w_1$\\ from an initial state to $q$} \qquad \text{and} \qquad
\myunderbrace{q \xrightarrow{w_2} F}{there is a run over $w_2$\\ from $q$ to a final state.}
\end{align*}
By \kl{unambiguity}, there is at most one intermediate state, and it exists if and only if the string is accepted.

Observe that Alice can compute the set of states that are reachable from an
initial state by reading her string $w_1$, and Bob can compute the set of
states from which a final state is reachable by reading his string $w_2$. So,
the challenge is to compute their intersection, which is either a singleton
with the intermediate state, or the empty set. Before we explain how to do
this, let us first explain why this is non-trivial, i.e. why Alice cannot send
her set of states to Bob, or vice versa. The problem is that the set of all
reachable subsets of $Q$ might not be \kl{orbit-finite}, and therefore it cannot be sent
in a constant number of messages. This issue is illustrated in the following
example.

\begin{myexample}
Consider the following unambiguous orbit-finite automaton for  the language ``the last letter appears at least twice''. The automaton guesses a position in the input string that is the penultimate appearance of the last letter, stores this letter in a register,
and verifies that its next appearance is indeed the last letter in the input word.
Formally, it has the following orbit-finite set of states:  
\begin{align*}
  Q \qquad = \qquad  \myunderbrace{1}{waiting for the\\ penultimate \\ appearance\\ of the last letter} \qquad +
      \qquad \myunderbrace{\atoms}{verifying that\\ the next appearance is\\ indeed the last letter}  \qquad + \qquad
      \myunderbrace{1}{next appearance of \\ the guessed letter; \\ accept if at \\ the end of input}.
\end{align*}
Observe that the automaton is unambiguous, as it takes care to accept only if the guessed position is indeed the penultimate appearance of the last letter.

Suppose now, that Alice and Bob want to simulate this automaton, and Alice has read a string $w_1 = a_1 \ldots a_n$ of pairwise distinct letters.
After reading $w_1$, the automaton can be in any  of the following $n+1$ states:
(a) waiting for the penultimate appearance of the last letter;
or (b) verifying that the next appearance of $a_i$ is the last letter, for $i \in \set{1,\ldots,n}$.
It follows that the size of the set of reachable states after reading $w_1$ is unbounded (as it grows with $n$),
and therefore it cannot be transmitted to Bob in a bounded number of orbit-finite messages.
\end{myexample}

To work around the issue explained above, Alice and Bob will engage in interactive communication, which will narrow down the set of possible candidates.
At each stage, it will be represented using orbits, as defined below.

\begin{definition}[Orbit] \label{def:orbit}
  \AP
  For a finite set $S \subseteq \atoms$, the \intro{$S$-orbit} of $q \in Q$ is the following set:
    \begin{align*}
    \setbuild{ \pi(q)}{$\pi$ is a permutation of $\atoms$ such that $\pi(a)=a$ for all $a \in S$}.
    \end{align*}
  The set $S$ is called the \intro{support} of the orbit.
\end{definition}

\begin{myexample}\label{ex:tau-disjoint}
    Let $Q = \atoms^5$, and consider the $\set{\text{{John}, {Eve}}}$-orbit of the following tuple
    \begin{center}
        (\red{John}, Tom, Mary, Tom, \red{Eve})
    \end{center}
    An element of this \kl{orbit} is any tuple of the form 
    \begin{center}
        (\red{John}, $a$, $b$, $a$, \red{Eve})
    \end{center}
    where $a$ and $b$ are distinct \kl{atoms}, which are not \red{John} or \red{Eve}. 
\end{myexample}

\AP As the \kl{support} increases, the \kl{orbit} becomes smaller; in
particular the biggest orbits are the ones with empty support, i.e. the
$\emptyset$-orbits, which we call the \intro{equivariant orbits}. It is not
hard to see that every \kl{orbit-finite} set has a finite number of
\kl{equivariant orbits}~\cite[Lemma 1.4]{bojanczyk_slightly}; in fact this is
the reason for the name. Each \kl{orbit} in an \kl{orbit-finite} set is a
subset of $\atoms^d$ for some $d$. In such an orbit, we partition the
coordinates $\set{1,\ldots,d}$ into two parts: the \intro{fixed coordinates},
which use the \kl{atoms} from the \kl{support}, and the \intro{free
coordinates}, which do not use these atoms. In Example~\ref{ex:tau-disjoint},
the \kl{fixed coordinates} are  the first and last ones, while the \kl{free
coordinates} are the middle three. The \intro[orbit dimension]{dimension} of an
orbit is the number of distinct atoms in the \kl{free coordinates}. In
Example~\ref{ex:tau-disjoint}, the \kl[orbit dimension]{dimension} is two,
corresponding to the atoms $a$ and $b$. An important special case is when the
\kl[orbit dimension]{dimension} is zero; in this case the orbit contains only one state.

In the \kl[of protocol]{protocol}, Alice and Bob will jointly maintain a set
$S \subseteq \atoms$ and list of \kl{$S$-orbits} which may contain the
intermediate state (starting with $S = \emptyset$ the list of all
$\emptyset$-orbits). They will ensure that if the intermediate state exists,
i.e.~if the input string is accepted,  then the intermediate state is contained
in one of the \kl{orbits} on the list. The goal is to decrease the \kl[orbit
dimension]{dimension} of the orbits on the list until they become
\kl{zero-dimensional}, by gradually computing the set $S$ of atoms that appear in
the intermediate state. Once the orbits become zero-dimensional, they will
contain only a finite (and bounded) number of the candidates. At this point,
Alice can compute which of these candidates are reachable from an initial state
over $w_1$ and send this (bounded) information to Bob, who can then check if
one of these candidates can reach a final state over $w_2$. 

To decrease the \kl[orbit dimension]{dimension} and increase $S$, we will use the following lemma.
\begin{lemma}\label{lem:fixed-atoms}
  Let $S \subset \atoms$ be a finite set, and let $\varphi \subseteq Q$ be an infinite \kl{$S$-orbit}
  of \kl[orbit dimension]{dimension} $k$. 
  Consider an input string $w = w_1 w_2$, and the sets:
  \begin{align*}
  \myunderbrace{X_1 = \setbuild{ q \in \varphi}{$ I \xrightarrow{w_1} q$}}{states reachable on Alice's side}
  \qquad
  \myunderbrace{X_2 = \setbuild{ q \in \varphi}{$ q \xrightarrow{w_2} F$}}{states reachable on Bob's side}
  \end{align*}
  There is a set $T \subseteq \atoms$ of size at most $k$, such that either: 
  \begin{enumerate}
    \item   every state from $X_1$ contains an \kl{atom} from $T$ on some \kl{free coordinate}; or 
    \item   every state from $X_2$ contains an \kl{atom} from $T$ on some \kl{free coordinate}.
  \end{enumerate}
\end{lemma}
    \begin{proof}
      \AP
 In the proof of the lemma, we use an analysis of disjointness, which is
 inspired by the sunflower lemma. We say that two states $p,q$ in an \kl{$S$-orbit}
 are \intro{$S$-disjoint} if 
 \begin{align*}
    \forall a \in \atoms 
    \qquad 
 \text{$a$ appears in both $p$ and $q$} \quad \Rightarrow \quad a \in S.
 \end{align*}
 In other words, the atoms from $S$ can appear in both $p$ and $q$ (in fact, they must), but all other atoms
must be disjoint in the following two states.
 For example, if we take the $\{\textrm{John}, \textrm{Eve}\}$-orbit from Example~\ref{ex:tau-disjoint}, then the two states
\begin{center}
    (\red{John}, Tom, Mary, Tom, \red{Eve}) \qquad
    (\red{John}, Ann, Timmy, Ann, \red{Eve})
\end{center}
are $\{\textrm{John}, \textrm{Eve}\}$-disjoint, because the sets
$\set{\text{Tom, Mary}}$ and $\set{\text{Ann, Timmy}}$ are disjoint. 

The following claim characterises subsets of orbits that do not contain any pair of disjoint elements:
\begin{claim}\label{claim:sunflower}
  Let $Q$ be an \kl{$S$-orbit} of \kl[orbit dimension]{dimension} $d$, and let $X$ be a subset of $Q$.
  If $X$ does not contain two \kl{$S$-disjoint} elements, then there is a
  set $T$ of at most $d$ atoms such that every element of $X$ uses at least one atom from $T$ on a \kl{free coordinate}.
\end{claim}
\begin{proof}
        Take some element $x \in X$. Either there is an element of $X$ that is completely disjoint with $x$, or otherwise some atom from $x$ must appear in every other element of $X$ on a free coordinate.
\end{proof}

The claim leaves us with showing that at least one of $X_1$ or $X_2$ does not
contain an \kl{$S$-disjoint} pair of elements. Suppose, towards a contradiction that
both $X_1$ and $X_2$ contain \kl{$S$-disjoint} pairs of elements, say $p_1, p_2 \in
X_1$ and $q_1, q_2 \in X_2$. It follows that the two pairs $(p_1,p_2)$ and
$(q_1,q_2)$ are in the same \kl{equivariant orbit} (of $Q \times Q$), i.e. there is
some atom permutation $\pi$ which sends $p_1$ to $q_1$ and $p_2$ to $q_2$.
Applying $\pi$ to Alices's part of the input string, we obtain a new input
string $\pi(w_1) w_2$, in which both $q_1$ and $q_2$ are valid intermediate
states. It follows that there are at least two accepting runs (one that passes
through $q_1$ and one that passes through $q_2$), contradicting the \kl{unambiguity}
assumption. 
\end{proof}

Using the above lemma, we will construct a \kl[of protocol]{protocol} that
simulates the automaton. As explained before, the idea is to narrow down orbit
which contains the intermediate state. This idea is formalised in the following
lemma.

\begin{lemma}\label{lem:narrow-down-orbit}
  Let $S$ be a finite subset of \kl{atoms}, and let $X \subseteq Q$ be an \kl{orbit}. 
  Alice and Bob can exchange a constant number of messages
  -- which depends only on the \kl[orbit dimension]{dimension} of $X$ -- 
  to determine if the intermediate state belongs to $X$. 
\end{lemma}

Before proving the lemma, let us explain how to use it to complete the proof of
\cref{thm:unambiguous-to-protocol}. We know that the set of all states $Q$
splits into constant number of \kl{equivariant orbits}, so the two parties can run
the protocol the lemma for each of these orbits with $S=\emptyset$. Each run of
the protocol uses a constant number of rounds, so the total number of rounds is
also constant. It remains to prove the lemma.

\begin{proof}[Proof of \cref{lem:narrow-down-orbit}]
  The proof proceeds by induction on the \kl{dimension of the orbit} $X$.
    
  The induction basis is when the dimension is zero. In this case, the orbit
  has exactly one state, and Alice and Bob can simply check  if the state is
  reachable on their side and exchange this bit of information.

  Consider now the induction step. Apply \cref{lem:fixed-atoms}, to the
  orbit. In the factorisation $w = w_1 w_2$, at least one of the two
  alternatives in the conclusion of \cref{lem:fixed-atoms} must hold. Alice
  can check if the first alternative holds, and Bob can check if the second
  alternative holds.  At least one of the two parties must report success,
  which is witnessed by some finite set $T$ of atoms. The successful party sends the
  set $T$ to the other party. This is possible since the size of $T$ is
  bounded by the \kl[orbit dimension]{dimension} of $X$. 
  The orbit $X$ splits into finitely many
  orbits $X_1,\ldots,X_n$ with the larger \kl[of support]{support} $S \cup T$, see~\cite[Lemma
  10.9]{bojanczyk_slightly}. The number $n$ depends only on the \kl[orbit dimension]{dimension} 
  of $X$ (as the size of $T$ is bounded by the dimension of $X$).

  We know that the intermediate state contains at least one atom from $T$, so
  we are only interested in the orbits among $X_1,\ldots,X_n$ which use at
  least one atom from $T$ on a coordinate that was \kl[free coordinate]{free} in $X$. 
  These orbits
  have lower \kl[of dimension]{dimension}, so the parties can sequentially apply the induction
  assumption to check if the intermediate state belongs to any of these
  orbits. This completes the proof of the lemma, and therefore also of
  \cref{thm:unambiguous-to-protocol}.
\end{proof}
\end{proof}

\subsection{Weighted automata}
\label{sec:weighted-automata-atoms}

\AP
In \cref{thm:unambiguous-to-protocol}, we have proved one implication in Conjecture~\ref{conj:protocols-unambiguous}.
This section is devoted to presenting some evidence for the other implication, i.e.
\begin{align}\label{eq:missing-orbit-finite-implication}
\text{protocol} \quad \implies \quad \text{unambiguous automaton}.
\end{align}
We begin by explaining why the orbit-finite case cannot be handled using the
techniques that were used to prove this implication in the finite case.

\paragraph*{What goes wrong in the orbit-finite case?}
In the finite case, the proof had two parts: (a) a reduction to \kl{one-round protocols}, and (b) the Myhill-Nerode Theorem.
Part (b) does not seem to be problematic, as orbit-finite versions of the Myhill-Nerode Theorem are known in many variants, 
including monoids~\cite[Lemma 3.3]{bojanczykNominalMonoids2013}, automata~\cite[Section 3.2]{bojanczykAutomataTheoryNominal2014}, 
and -- as we will prove later in this section -- also for weighted automata. 
The problematic part is (a), in which the number of rounds is reduced to one.
The key argument in this reduction  was that the sets of strategies 
  \begin{align*}
    (Q_B)^k \to (Q_A)^k \qquad \text{and} \qquad (Q_A)^k \to (Q_B)^k
    \end{align*}
are finite, and thus each party could simply send their strategy as a message.
This argument fails to carry over from finite sets to orbit-finite sets. 
The reason is that \kl{orbit-finite} sets are not closed under taking function spaces $X \to Y$, see~\cite{functionSpaces2024} 
for an extended discussion of this phenomenon.
The following example shows that the one-round reduction 
is indeed impossible in the orbit-finite case.

\begin{myexample}
    [No reduction to one round]\label{ex:no-one-round-reduction}
    Consider a language $L$ that is computed by an \kl{orbit-finite protocol} with
    one round. Using the same argument as in
    \cref{lem:one-round-reduction-boolean}, we can show that the Myhill-Nerode
    equivalence relation for the language, as defined
    in~\eqref{eq:myhill-nerode-equivalence}, has an \kl{orbit-finite} set of
    equivalence classes. As mentioned above, it follows from \cite[Section
    3.2]{bojanczykAutomataTheoryNominal2014} that $L$ is also recognised by a
    \kl{deterministic orbit-finite automaton}. As we have seen in
    Example~\ref{ex:reg-det-too-weak}, such automata are not strong enough to
    capture all \kl[of protocol]{protocols}.
\end{myexample}

In light of the above example, it is no longer surprising that the proof of
\cref{thm:unambiguous-to-protocol} used multi-round protocols. In fact, we
believe that the number of needed rounds can be arbitrarily large, as suggested
by the following example. 

\begin{myexample}[Many rounds]
    We think of a string over the alphabet $\atoms^2$ as a representation of  a partial function of type $\atoms \to \atoms$. The general idea is that we view the input string as consisting of a list of (input, output) pairs, and if an input appears several times, then the leftmost occurrence is binding. More formally, the  function is defined as follows: if the input is $a$, then we search for the leftmost letter $(a,b)$ that appears in the input string and has $a$ as the first component. If there is such a letter, then the output is $b$, otherwise the ouptut is undefined. Using this representation, we define for each  $k \in \set{0,1,\ldots}$ a  language as follows:
    \begin{align*}
    L_k = \setbuild{ w \in (\atoms^2)^* }{ if $f$ is the partial function represented by $w$, then \\  $f^k(a)=b$ where $(a,b)$ is the last letter of $w$ }.
    \end{align*}
     This language can be computed by an orbit-finite protocol with $k$ rounds, with each round corresponding to one iteration of the function. It seems unlikely that any bounded number of rounds would suffice for all $L_k$, but we do not prove this claim here.
\end{myexample}

\paragraph*{Weighted orbit-finite automata.} In the reminder of this section,
we present some evidence for the missing implication in
Conjecture~\ref{conj:protocols-unambiguous}, using the orbit-finite version of
weighted automata. Namely, we will prove
\cref{thm:orbit-finite-protocol-to-weighted}, which states that every language
computed by an \kl{orbit-finite protocol} is also recognised by a \kl{weighted
orbit-finite automaton} over the two-element field. For finite alphabets, this
would be enough to ensure \kl[regular language]{regularity}, see
\cref{claim:regular-weighted-automata}. This claim is no longer true in the
orbit-finite case, see Example~\ref{ex:weighted-vs-nondet-orbit-finite}, and
therefore \cref{thm:orbit-finite-protocol-to-weighted} can only be considered
as evidence for the conjecture. However, at the very least it shows that
languages computed by orbit-finite protocols are decidable, which was not a
priori clear from the definition.

Let us begin by defining the orbit-finite version of weighted automata. 
\begin{definition}[Weighted orbit-finite automata]
    \label{def:weighted-orbit-finite-automata}
    \AP
    A \intro{weighted orbit-finite automaton} over a semiring $\domain$ is 
    defined in the same way as in Definition~\ref{def:weighted-automaton-nondeterministic}, except that:
    \begin{enumerate}
      \item the input alphabet and state space are \kl{orbit-finite}, instead of finite;
      \item the functions in~\eqref{eq:weight-functions} are \kl{equivariant}.
    \end{enumerate}
     We require that for every input string, there are only finitely many runs with non-zero weight.
\end{definition}

For the purpose of this section, already the special case of the two-element
field $\set{0,1}$ is interesting. In this case, the automaton defines a
function $\Sigma^* \to \set{0,1}$, which can be seen as the characteristic
function of a language. Therefore, we can compare \kl{weighted orbit-finite
automata} to other models, such as \kl{nondeterministic orbit-finite automata}. The
following example shows that these two  models are incomparable. 

\begin{myexample}\label{ex:weighted-vs-nondet-orbit-finite}
  The language ``some letter appears twice'' is recognised by a
  \kl{nondeterministic orbit-finite automaton}, but its characteristic function
  cannot be  recognised by a \kl{weighted orbit-finite automaton} over the
  two-element field. The non-expressivity can be proved using the orbit-finite
  version of the Fliess Theorem, see \cref{thm:orbit-finite-fliess}. On the
  other hand, the  language ``an even number of distinct letters'' is not
  recognised by a nondeterministic orbit-finite automaton, while its
  characteristic function can be computed by a \kl{weighted orbit-finite
  automaton}, see~\cite[Example 3.2]{orbitFiniteVectorTheoretics}. Therefore,
  the two models -- nondeterministic and weighted in the two-element field --
  are incomparable. Inside the intersection of these two classes we will find
  the \kl[of unambiguous automata]{unambiguous automata}, since for unambiguous automata, counting the runs
  modulo two gives the same result as checking if a run exists. This discussion
  is summed up in the following picture:
    \mypic{2}
\end{myexample}

The following theorem is the main result of Section~\ref{sec:weighted-automata-atoms}.
\begin{theorem}\label{thm:orbit-finite-protocol-to-weighted}
  Let $\Sigma$ be an \kl{orbit-finite} input alphabet, and let $\domain$ be a field.
  If a language $L \subseteq \Sigma^*$ is computed by a \kl[of protocol]{protocol}, 
  then the corresponding characteristic function of type $\Sigma^* \to \set{0,1} \subseteq \domain$
  is computed by a \kl{weighted orbit-finite automaton}.
\end{theorem}

In the proof of the theorem, we use the recently developed theory of
orbit-finite vector spaces \cite{orbitFiniteVectorTheoretics}. To streamline
the presentation, we will use a special case of these spaces, namely those that
have an orbit-finite basis. We begin by summarizing the necessary background:
For an orbit-finite set $Q$, let us write $\lincomb Q$ for the vector space
which consists of finite formal linear combinations of elements of $Q$. In
other words, an element of this space is a vector of the form 
\begin{align*}
\alpha_1 q_1 + \cdots + \alpha_n q_n,
\end{align*}
where the coefficients $\alpha_i$ are from the field, and the element $q_i$ (which can be seen as basis vectors) are from $Q$. We use such spaces as the orbit-finite generalisation of finite dimension\footnote{Similarly to the case of orbit-finite sets, we use a simplified definition for this paper, as compared to the literature. The usual notion of vector spaces for orbit-finite sets, see \cite[Definition 8.1]{bojanczyk_slightly} allows for more features; these features are irrelevant for our application and hence not discussed here.}.

\begin{definition}
    [Vector space of orbit-finite dimension]\label{def:orbit-finite-vector-space}
    \AP
    A \intro{vector space of orbit-finite dimension} 
    is a vector space of the form $\lincomb Q$ for some \kl{orbit-finite} set $Q$.
\end{definition}
A space as in the above definition is equipped with two structures: as a vector
space it is closed under linear combinations, and as a set with \kl{atoms} it
is closed under applying \kl{atom permutations}. We will typically be
interested in functions between such spaces that preserve both of those
structures.

In the proof of \cref{thm:orbit-finite-protocol-to-weighted}, we will use an
orbit-finite version of \kl{protocols} with field outputs. Recall that there
are two variants: the general version from Section~\ref{sec:intro-field} and
the simpler \kl{scalar-product protocols} from Section~\ref{sec:field-domain}.
We could begin with the general version and show that it is equivalent to the
scalar-product one -- this is indeed the case. However, to keep the exposition
concise, we treat \kl{orbit-finite protocols} with field outputs merely as a tool
for proving \cref{thm:orbit-finite-protocol-to-weighted}, rather than as an
object of independent interest. Therefore, we focus on the simplest form of
protocol required for the proof, namely an orbit-finite version of the scalar
product protocols from Section~\ref{sec:field-domain}. In the orbit-finite
case, instead of scalar products we will use the slightly more general notion
of bilinear maps, i.e.~maps that have two arguments and are linear in each
argument separately\footnote{It is possible to define orbit-finite scalar
  product protocols, but they are harder to work with. In particular, from our
  results it will follow that orbit-finite scalar product protocols (suitably
  defined) and orbit-finite bilinear protocols are equivalent, but we are not
  aware of a direct proof of this fact.}. 


\begin{definition}
    [Orbit-finite bilinear protocol] 
    \label{def:orbit-finite-scalar-product-protocol}
    \AP
    An \intro{orbit-finite bilinear protocol} consists of:
    \begin{enumerate}
      \item two \kl{vector spaces of orbit-finite dimension}
        $V_A$ and $V_B$ and two strategies, which are equivariant functions
        \begin{align*}
        \sigma_A : \Sigma^* \to V_A 
        \quad \text{and} \quad
        \sigma_B : \Sigma^* \to V_B
        \end{align*}
        \item an output map, which is an equivariant bilinear map 
        \begin{align*}
        \text{out} : V_A \times V_B \to \domain.
        \end{align*}
    \end{enumerate}
\end{definition}
The output of the \kl[of bilinear protocol]{protocol} is defined in the natural way: 
Alice and Bob apply their strategies to their local strings, yielding two vectors, 
and the output of the protocol is obtained using the output map. 
As usual, we require split invariance, i.e.~the output of the protocol should
depend only on the input string $w$ and not on its factorisation $w = w_1 w_2$.


Let us go back to the proof of \cref{thm:orbit-finite-protocol-to-weighted}. The proof has two steps, as described in the following diagram.
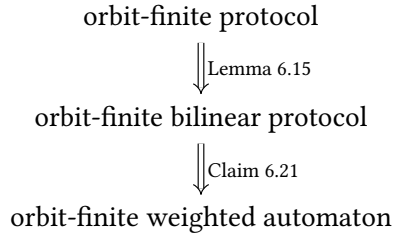
\[
\begin{tikzcd}
\text{orbit-finite protocol}
\ar[d,Rightarrow,"\text{Lemma~\ref{lem:orbit-finite-protocol-to-scalar}}"]
\\
\text{orbit-finite bilinear protocol}
\ar[d,Rightarrow, "\text{\cref{claim:bilinear-prot-to-of-automaton}}"]
\\
\text{orbit-finite weighted automaton}
\end{tikzcd}
\]

We begin with the first step, which can be seen as form of reduction to one round. Recall that without vector spaces, a reduction to one round was not possible, see Example~\ref{ex:no-one-round-reduction}. This phenomenon is connected to closure under taking function spaces: orbit-finite sets are not closed under taking function spaces, but this closure is recovered once one moves to vector spaces, see~\cite[Section 8.3]{bojanczyk_slightly}. 

\begin{lemma}\label{lem:orbit-finite-protocol-to-scalar}
  If $L \subseteq \Sigma^*$ is computed by an \kl{orbit-finite protocol}, then its characteristic function is computed by an \kl{orbit-finite bilinear product protocol}, over any field.
\end{lemma}

\begin{proof}
  \AP
    We first introduce a common generalisation of the two models called \intro{hybrid protocols}. 
    In such a protocol, there are multiple rounds of messages, 
    followed by a bilinear operation. We will show that the multiple rounds can be eliminated,
    yielding a plain \kl{orbit-finite bilinear protocol}.
    
    Here is the formal definition of the \kl{hybrid protocol}:
    For each round $i \in \set{1, \ldots, k-1}$, the parties exchange messages just as in an 
    \kl[of protocol]{orbit-finite protocol} from Definition~\ref{def:orbit-finite-protocol}, 
    using message spaces $Q_A$ and $Q_B$ and strategies of the following types: 
        \begin{align*}
        \sigma_{A,i} & : \Sigma^* \times (Q_B)^{i-1} \to Q_A\\
        \sigma_{B,i} & : \Sigma^* \times (Q_A)^{i-1} \to Q_B
        \end{align*}
    Then, in  the last $k$-th round, the message histories are used to produce vectors 
    in two \kl[ofd vector spaces]{vector spaces} $V_A$ and $V_B$ of orbit-finite dimension, 
    as in Definition~\ref{def:orbit-finite-vector-space}, using strategies 
    of types
    \begin{align*}
        \sigma_{A,k} & : \Sigma^* \times (Q_B)^{k-1} \to V_A\\
        \sigma_{B,k} & : \Sigma^* \times (Q_A)^{k-1} \to V_B.
        \end{align*}
    Finally, from the two vectors, the output is computed using a bilinear map
    \begin{align*}
        \text{out} : V_A \times V_B \to \domain.
    \end{align*}

    The hybrid protocol generalises both \kl{orbit-finite protocols} and 
    \kl{orbit-finite bilinear protocols}. 
    For the latter, this is clear: we simply use $k=1$ and there is no message exchange. 
    For the former, we proceed as follows.
    We use  trivial vector spaces, i.e.~both $V_A$ and $V_B$ are the field. 
    The bilinear map is multiplication. 
    Once the two parties have agreed on a Boolean decision, 
    they can both send $1$ (in the case of a ``yes'' decision) or $0$ (in the case of a ``no'' decision), 
    and the bilinear map will give the correct output. 
    
    In order to complete the proof of the lemma, we will show that the number
    of rounds can always be reduced to one, thus yielding a \kl[of bilinear protocol]{bilinear protocol}.

    \begin{claim}\label{claim:reduce-round}
      For every $k > 1$, a \kl{hybrid protocol} with $k$ rounds can be simulated by a 
      \kl{hybrid protocol} with $k-1$ rounds.
    \end{claim}
    \begin{proof}
        We will eliminate round $k-1$, where the last message is sent. 
        Once Alice has received the first $k-2$ messages from Bob, 
        her contribution to the rest of the protocol is described by an object of type 
          \begin{align}\label{eq:contribution-last-two-rounds}
            \myunderbrace{Q_A}{message \\ sent in \\ round $k-1$} \quad \times \quad  \myunderbrace{(\fsfun  {Q_B} {V_A})}{message sent in \\ round $k$, as a function \\ of the message sent \\ in  round $k-1$}
        \end{align}

        Before we continue, let us now explain the ``fs'' annotation on the arrow above. It stands for
        the set of \emph{finitely supported functions}, which is a relaxation of
        equivariant functions. The general idea is that while equivariant functions can only refer to equality between letters, the  finitely supported functions can also use finitely many constants. Before presenting the definition, we give two examples. The first example is the function of type $\atoms \to \atoms$ defined by 
        \begin{align}\label{eq:fs-function-example}
        a \in \atoms \mapsto 
        \begin{cases}
        \text{John} & \text{if $a = \text{Eve}$}\\
        a & \text{otherwise}.
        \end{cases}
        \end{align} 
        This function is  finitely supported but not equivariant. Here is a second example of finitely supported functions, which is particularly  relevant to our application.

        \begin{myexample}[Partial application not equivariant]
        \label{ex:partial-application-not-equivariant}
        A partial application of a function 
        \begin{align*}
        f : X \times Y \to Z
        \end{align*}
        to an argument $x \in X$ is a function $f(x, -) : Y \to Z$, defined as $f(x, -)(y) = f(x,y)$. Even if $f$ is equivariant, then the  partially applied function $f(x, -)$ might not be equivariant. For example, consider the equivariant function
        \begin{align*}
        f: {\atoms^* \times \atoms} \to  {\set{\textrm{true}, \textrm{false}}},
        \end{align*}
        such that $f(w, a) = \textrm{true}$ if and only if $a$ appears in $w$. Now, consider the partial application
        \begin{align*}
        f' := f(\textrm{Celine}\, \textrm{John}\, \textrm{Anne}, -)
        \end{align*}
        which is of type $\atoms \to \set{\textrm{true}, \textrm{false}}$. This function is no longer equivariant, because the input needs to be compared to the constants Celine, John, and Anne. However, the number of constants is finite, hence it is finitely supported.
      \end{myexample}
      \begin{definition} [finitely supported function]
        \label{def:finitely supported-function}
        \AP
         A function $f : X \to Y$ is \intro{finitely supported} if there is a finite set $S$ of atoms such that for every atom permutation $\pi$ we have 
        \begin{align*}
            (\forall a \in S\ \pi(a)=a) 
            \quad \implies \quad f(\pi(x)) = \pi(f(x)).
        \end{align*}
      \end{definition}
      An important property of \kl{finitely supported} function spaces is that they do not preserve orbit-finiteness: 

      \begin{myexample} If we take two orbit-finite sets, then the set of \kl{finitely supported} functions between them might not be \kl{orbit-finite}. For example, the set
        \begin{align*}
        \fsfun{\atoms}{\set{\text{true, false}}}
        \end{align*}
        of \kl{finitely supported} functions from $\atoms$ to the Booleans is not \kl{orbit-finite}, since such functions can refer to any number of constants, as demonstrated by the partial applications in Example~\ref{ex:partial-application-not-equivariant}.
      \end{myexample}

      As mentioned before, one can recover closure under taking function spaces by moving to vector spaces. This is formalized in the following lemma, which stems 
      from \cite{orbitFiniteVectorTheoretics}, but we use its formulation from \cite{bojanczyk_slightly}, as it is more convenient for our purposes:
      \begin{lemma}[{\cite[Lemma 8.17]{bojanczyk_slightly}}]
      \label{lem:function-space-of-ofd}
      For all \kl{orbit-finite} sets $X$ and $Y$, the space of \kl{finitely supported} functions $\fsfun{X}{\lincomb Y}$ is a \kl{vector space of orbit-finite dimension}.
      \end{lemma}
      Let us now close this digression about function spaces (for more discussion, see \cite[Section 8.3]{bojanczyk_slightly} or \cite{functionSpaces2024})
      and go back to the proof of the claim. Observe, that the second coordinate in~\eqref{eq:contribution-last-two-rounds} is indeed of type 
      \begin{align*}
      \fsfun{Q_B}{V_A},
      \end{align*}
      as it arises from a partial application of the equivariant function $\sigma_{A,k}$ to Alice's part of the input (i.e. $w_1 \in \Sigma^*$)
      and her message history from the first $k-2$ rounds (i.e. an element of $(Q_B)^{k-2}$). Therefore, in order to transmit it to Bob, 
      we need to turn the type in~\eqref{eq:contribution-last-two-rounds} into a vector space. The second coordinate is already a vector space,
      since functions with outputs in a vector space can be added and scaled pointwise. What is more, by Lemma \ref{lem:function-space-of-ofd},
      it has an orbit-finite basis:
      \begin{align*}
        F_A \subseteq \fsfun  {Q_B} {V_A}.  
      \end{align*}
      The first coordinate $Q_A$ in~\eqref{eq:contribution-last-two-rounds} can be turned into a vector space by  allowing linear combinations, i.e.~$\lincomb Q_A$. We combine these two using tensor product, yielding a vector space of orbit-finite dimension
        \begin{align*}
           W_A =  (\lincomb Q_A) \otimes (\lincomb F_A).
        \end{align*}
        We can do the same thing for Bob, obtaining a vector space
        \begin{align*}
           W_B =  (\lincomb Q_B) \otimes (\lincomb F_B),
        \end{align*}
        where $F_B$ is an orbit-finite basis of the vector space 
        \begin{align*}
        \fsfun  {Q_A} {V_B}.
        \end{align*}
        In order to define an equivariant bilinear map of type
        \begin{align*}
        \varphi : W_A \times W_B \to \domain
        \end{align*}
        it is enough to define an equivariant linear map of type 
        \begin{align*}
        \varphi : W_A \otimes W_B \to \domain
        \end{align*}
        for which it is enough to define an equivariant function on its basis
        \begin{align*}
        Q_A \times F_B \times Q_B \times F_A \to \domain
        \end{align*}
        This definition is the only one that types, namely 
        \begin{align*}
        (q_A, f_B, q_B, f_A) \quad 
        \mapsto \quad 
        \text{out}(f_A(q_B), f_B(q_A)).
        \end{align*}
        Because the output map is bilinear, one can check that $\varphi$ defined this way is consistent with the original protocol, i.e.~if we take functions 
        \begin{align*}
        f_A : Q_B \to V_A \qquad \text{and} \qquad f_B : Q_A \to V_B,
        \end{align*}
        which are not necessarily basis vectors from $F_A$ and $F_B$, then we have 
        \begin{align*}
        \text{out}(f_A(q_B), f_B(q_A)) = \varphi((q_A, f_B), (q_B, f_A)).
        \end{align*}
        Therefore, we can implement the last two round of the original hybrid protocol using a single round. The message spaces and the strategies for the first $k-2$ rounds are unchanged. In the last round $k-1$, the new strategies
        \begin{align*}
        \sigma'_{A,k-1} & : \Sigma^* \times (Q_B)^{k-2} \to W_A\\
        \sigma'_{B,k-1} & : \Sigma^* \times (Q_A)^{k-2} \to W_B
        \end{align*}
        output the tensor pairs consisting of the contribution that was described in~\eqref{eq:contribution-last-two-rounds}. Finally, the output  map for the new protocol is $\varphi$. 
    \end{proof}

    By repeatedly applying the above claim, we can reduce the number of rounds to one, in which case we get a bilinear protocol, as required in the statement of the lemma. 
\end{proof}

\subsubsection{Orbit-Finite Fliess Theorem}
In this section, we prove an orbit-finite version of the Fliess Theorem, which characterises functions $\Sigma^* \to \domain$ that are computed by \kl{weighted orbit-finite automata}. 
This result will be used to complete the proof of \cref{thm:orbit-finite-protocol-to-weighted}.

As in the original Fliess Theorem, we will be interested in \kl{derivatives} of the function,
which live in the space  
\begin{align*}
\Sigma^* \to \domain.
\end{align*}
This space has three kinds of structure, all of which will are used in the Fliess Theorem: 
\begin{enumerate}
    \item It is a vector space, since we can take linear combinations of functions.
    \item It has a notion of \kl{left derivatives}, i.e.~for each function $f$ and input string $w \in \Sigma^*$, we can consider the \kl{left derivative} $v \mapsto f(wv)$, which is denoted by $\leftderivative{f}{w}$.
    \item It has a notion of \kl{atom permutations}: for each function $f$ and atom permutation $\pi$, we can consider the function $\pi(f)$, which is the composition $\pi;f$.
\end{enumerate}

\AP
We say that a subset  $U \subseteq \Sigma^* \to \domain$  is
\intro{orbit-finitely spanned} if there is some \kl{orbit-finite} set $Q$, such that
every element of $U$ is a finite linear combination of elements from $Q$. We do
not require the linear combination to be unique, i.e.~we do not require $Q$ to
be a basis. (Choosing a basis can be problematic in the context of orbit-finite
sets, see~\cite[Example 77]{bojanczyk_slightly}.) We are now ready to state the
orbit-finite version of the Fliess Theorem.

\begin{theorem}[Orbit-Finite Fliess Theorem]\label{thm:orbit-finite-fliess}
    The following two conditions are equivalent for every function
    $f : \Sigma^* \to \domain$
    where $\Sigma$ is an orbit-finite alphabet and $\domain$ is a field.
    \begin{enumerate}
      \item \label{it:fliess-weighted} $f$ is computed by a \kl{weighted orbit-finite automaton};
      \item \label{it:fliess-derivatives} $f$ is \kl{equivariant} and its set of
        \kl{left derivatives} is \kl{orbit-finitely spanned}.
    \end{enumerate}
\end{theorem}
\begin{proof} Our proof follows the lines of the original theorem, without any significant changes.

    \AP
    We begin with the implication \ref{it:fliess-weighted} $\implies$ \ref{it:fliess-derivatives}. Consider a weighted orbit-finite automaton with state space $Q$. 
    Define the \intro{pre-weight} of a run in the same way as its weight, except that we do not use the final weight. In other words, this is the product of: (1) the initial weight of the first state; and (2) the weights of all transitions. Consider an input string $w$. Define the \intro{configuration} of $w$ to be the linear combination
    \begin{align}
        \label{eq:configuration-wa}
        \sum_\rho \alpha_\rho \cdot q_\rho,
    \end{align}
    where $\rho$ ranges over runs that have input $w$ and non-zero \kl{pre-weight},
    $\alpha_\rho$ is the \kl{pre-weight} of the run $\rho$, and $q_\rho$ is the last state in this run.
    By the assumption that each input string has finitely many runs with non-zero weight, 
    the configuration is a finite sum, i.e.~it belongs to the vector space $\lincomb Q$. 
    The \kl{left derivative} which corresponds to the input string is uniquely determined 
    by this configuration, and the space of \kl{configurations} is \kl{orbit-finitely spanned}. 
    Hence, we get~\ref{it:fliess-derivatives}.

    We now prove the other implication, \ref{it:fliess-derivatives} $\implies$
    \ref{it:fliess-weighted}. Assume~\ref{it:fliess-weighted}, which means that
    there is an \kl{orbit-finite} set $Q \subseteq \Sigma^*$ such that every
    \kl{derivative} of $f$ can be decomposed as a finite linear combination of left
    derivatives
    \begin{align*}
            \sum_i \alpha_i \leftderivative{f}{w_i},
    \end{align*}
    where each string $w_i$ is in $Q$. The following claim shows that strings obtained by taking elements of $Q$ and appending one letter
    (i.e. strings of the form $Q \cdot \Sigma$) can be decomposed in an equivariant way:
    \begin{claim}
      There is an \kl{equivariant function}
        \begin{align*}
       \delta :  Q \times \Sigma \to \lincomb Q
        \end{align*}
        such that the following conditions holds for every $w \in Q$ and $a \in \Sigma$:
        \begin{align*}
        \delta(w,a) = 
        \sum_i \alpha_i w_i 
        \qquad \Rightarrow \qquad 
        \leftderivative f {wa} = \sum_i \alpha_i \leftderivative{f}{w_i}.
        \end{align*}
    \end{claim}
    \begin{proof}
        Observe that Condition~\ref{it:fliess-derivatives} in the theorem's statement already asserts that there exists such a function $\delta$,
        possibly non-equivariant. In this proof we show how to modify it so that it becomes equivariant
        while still satisfying other requirements of the claim. We construct this modified function $\delta'$ as follows:
        for every orbit $Q \times \Sigma$ pick its representative $(w, a)$ and then extend the result to the whole orbit by equivariance:
        \begin{align*}
        \delta'(\pi(w), \pi(a)) := \pi(\delta(w, a))
        \end{align*}
        The new function $\delta'$ is equivariant by construction, and it keeps satisfying other requirements of the claim because
        $f$ is an equivariant function, so its derivatives commute with atom permutation:
        \begin{align*}
        \leftderivative{f}{\pi(w)}= \pi(\leftderivative{f}{w}).
        \end{align*}
        This gives  us that:
        \begin{align*}
                    \leftderivative{f}{wa} = \sum_i \alpha_i \leftderivative{f}{w_i}
                    \quad \iff \quad 
                            \leftderivative{f}{\pi(wa)} = \sum_i \alpha_i \leftderivative{f}{\pi(w_i)},
        \end{align*}
        which completes the proof of the claim.
    \end{proof}        

    Using the function $\delta$ from the above claim, we define a \kl{weighted
    orbit-finite automaton}. The state space is the set $Q$. (We assume without
    loss of generality that $Q$ contains the empty string $\varepsilon$. This
    is not really necessary for the construction, but it makes it more
    intuitive.) The weights are defined as follows: 
    \begin{itemize}
        \item \textbf{Initial weights.} The initial weight of   $\varepsilon$ is $1$. All other states have initial weight $0$.
        \item \textbf{Transition weights.} The weight of a transition 
        \begin{align*}
        w \xrightarrow{a} v
        \end{align*}
    is the coefficient next $v$ in the linear decomposition $\delta(w,a)$.
        \item \textbf{Final weights.} The final weight of a state $w \in Q$ is $f(w)$.
    \end{itemize}
    Remember that \kl{orbit-finite weighted automata} can only admit finitely many runs for each input word.
    Our automaton satisfies this requirement, as all linear combinations returned by $\delta$ are finite. 

    Finally, we justify why this automaton computes the function $f$. A simple inductive proof shows that  
    \begin{align*}
\text{configuration of $w$} = \sum_\rho \alpha_\rho \cdot w_\rho
\quad \implies \quad 
        \leftderivative f w = 
    \sum_\rho \alpha_\rho \cdot \leftderivative f {w_\rho}.
    \end{align*}
    By choice of final weights, the output of the automaton is equal to $f(w)$. 
\end{proof}

We are now ready to complete the proof of \cref{thm:orbit-finite-protocol-to-weighted} by showing the missing link:
\begin{claim}
\label{claim:bilinear-prot-to-of-automaton}
Let $\domain$ be a field.
    If a function $\Sigma^* \to \domain$ is computed by an 
    \kl{orbit-finite bilinear protocol}, it can also be computed by an 
    \kl{orbit-finite weighted automaton}.
\end{claim}
\begin{proof}
    Thanks to the Orbit-Finite Fliess Theorem,
    it is enough to show that if a function is computed by an \kl{orbit-finite bilinear 
    protocol}, then it has an \kl{orbit-finitely spanned vector space} of \kl{left derivatives}. This follows from the same argument as in Section~\ref{sec:field-domain}:
    the vector produced by Alice in a \kl[of bilinear protocol]{bilinear protocol}
    uniquely determines the \kl{left derivative} of her part of the input. 
\end{proof}

\section{Conclusions}
\label{sec:conclusions}
One could possibly consider other inputs, such as trees. It seems that at least some of our results could generalise to such inputs, as long as there would be a suitable theory of automata for the inputs, with theorems in the style of Myhill-Nerode (which is the case for trees). Nevertheless, we leave the the exploration of non-string inputs to future work.

\bibliographystyle{alpha}
\bibliography{bib}

\end{document}